\documentclass[11pt]{article}

\pdfoutput=1

\usepackage[margin=3cm]{geometry}
\usepackage{graphicx}
\usepackage{amssymb}	
\usepackage{amsmath}
\usepackage{amsthm}
\usepackage{enumerate}
\usepackage[mathscr]{euscript}
\usepackage{xcolor}

\usepackage{fancyhdr}
\newtheorem{lemma}{Lemma}
\newtheorem{definition}{Definition}
\newtheorem{corollary}{Corollary}
\newtheorem{theorem}{Theorem}
\newtheorem{observation}{Observation}

\newcommand{\deleted}[1]{}

\newcommand{\etal}{{~et~al.~}}

\newcommand{\RR}{\mathbb{R}^2}

\newcommand{\fvd}{\mathrm{FVD}}
\newcommand{\J}{\mathcal{J}} 		
\newcommand{\V}{\mathcal{V}} 
\newcommand{\Vk}{\V_k} 
\newcommand{\VR}{\mbox{VR}}    	
    
\newcommand{\sms}{S \setminus \{s\}} 
 
\newcommand{\arcs}{\mathscr{S}}
\newcommand{\subs}{\arcs'}
\newcommand{\arr}{\mathcal{A}}
\newcommand{\fs}[1]{\text{VR}(#1, \arcs)}
\newcommand{\vld}{\mathcal{V}_l}
\newcommand{\ins}{{\vld(\pp)\oplus\beta}}

\newcommand{\pp}{\mathcal{P}}
\newcommand{\env}{\text{env}}
\newcommand{\E}{\mathcal{E}}
\newcommand{\jj}[2]{\J_{#1,#2}}
\newcommand{\freg}{\mbox{FVR}} 

\newcommand{\kreg}[1]{\mbox{VR}_k{#1}}

\newcommand{\T}{\mathcal{T}}
\newcommand{\N}{N}

\newcommand{\In}{\text{in}}
\newcommand{\out}{\text{out}}
 
\newcommand{\so}{\text{source}}

\newcommand{\B}{\mathcal{B}}

\newcommand{\floor}[1]{\left\lfloor #1 \right\rfloor}
\newcommand{\aj}{\alpha}
\newcommand{\gj}{\gamma}
\newcommand{\f}{f}

\newenvironment{apptheorem}[1]{\vspace{5pt}\noindent
	\textcolor{darkgray}{$\blacktriangleright$}\nobreakspace\sffamily\bfseries\textbf{Theorem #1.}\it}{\vspace{2pt}}

\begin{document}

\title{Deletion in abstract Voronoi diagrams in expected linear time
  and related problems 
  \thanks{This research was supported 
     in part by the Swiss National Science Foundation, project
     200021E$\_$154387.
   }
   \thanks{A preliminary version of this paper appeared in \emph{Proc.  34th
    International Symposium on Computational Geometry (SoCG)
    2018}.
 }
}

\deleted{
\titlerunning{Deletion in abstract Voronoi diagrams}

\author{Kolja Junginger \and Evanthia Papadopoulou}

\institute{Kolja Junginger \at
              Faculty of Informatics, USI Universit\`a
              della Svizzera italiana, \\{Lugano, Switzerland} \\
              \email{kolja.junginger@usi.ch}           
           \and
           Evanthia Papadopoulou \at
       Faculty of Informatics, USI
       Universit\`a della Svizzera italiana, \\
       {Lugano, Switzerland}\\
	\email{evanthia.papadopoulou@usi.ch}
}
}
\date{}


\author{\bigskip Kolja Junginger, Evanthia Papadopoulou\\
  Faculty of Informatics, USI
  Universit\`{a} della Svizzera italiana,\\
  {Lugano, Switzerland}\\
  \texttt{kolja.junginger@usi.ch},
  \texttt{evanthia.papadopoulou@usi.ch}
\smallskip}

\maketitle

\begin{abstract}
\noindent
Updating an abstract Voronoi diagram after deletion
of one site in linear time  
has been a well-known open problem;
similarly, for concrete Voronoi diagrams of 
non-point sites. 
In this paper, we present an expected linear-time algorithm to 
update an abstract Voronoi diagram after deletion of one site.
We introduce the concept of a \emph{Voronoi-like diagram}, a relaxed 
version of an abstract Voronoi construct that has a structure similar
to an ordinary Voronoi diagram, without, however, being one.
We formalize the concept, and prove that it is robust under insertion, therefore, 
enabling its use in incremental constructions.
The time-complexity analysis
of the resulting simple randomized incremental construction
is non-standard, and interesting in its own right, because the intermediate Voronoi-like structures are
order-dependent.
We further extend
the approach to compute the following structures in expected linear time: the order-$(k{+}1)$ subdivision within an order-$k$
Voronoi region, and
the farthest abstract Voronoi diagram after the order of its regions
at infinity is known.
\deleted{
The time-complexity analysis of the randomized incremental
construction is non-standard
and interesting in its own right, as the intermediate Voronoi-like structures are
order-dependent.
}

\paragraph{Keywords:} 
abstract Voronoi diagram; linear-time algorithm; randomized incremental construction; 
site-deletion; higher-order Voronoi diagram; farthest Voronoi diagram. 
\medskip
\end{abstract}

\section{Introduction} \label{sec:intro}
The Voronoi diagram of a set $S$ of $n$ simple geometric objects,
called
sites, is a versatile and well-known geometric partitioning
structure, 
which reveals proximity information for the input sites.
Classic variants include the \emph{nearest-neighbor},
the \emph{farthest-site}, and the \emph{order-$k$} Voronoi diagram
of the set $S$. 
Abstract Voronoi diagrams offer a unifying framework
to many  concrete and fundamental instances. 
Voronoi diagrams  have been well-investigated and 
optimal construction algorithms exist in many cases.
For more information
see, e.g., the book of Aurenhammer\etal\cite{Aurenbook} and
also~\cite{VoronoiBook00} for 
a wealth of applications.
\deleted{
For references and information
see, e.g., the book of Aurenhammer\etal\cite{Aurenbook};
see also~\cite{VoronoiBook00} for 
a wealth of applications.
}

For certain Voronoi diagrams with a tree
structure, linear-time construction algorithms
have been well-known to exist, see e.g.,~\cite{AGSS89,Chew90,KL94,snoyeink99}.
The first
linear-time technique was introduced by Aggarwal\etal\cite{AGSS89} 
for  the Voronoi diagram of points in convex
position, given the order of the points along their convex hull.
The same technique can be used to derive linear-time 
algorithms for other fundamental problems:  
(1) updating a Voronoi diagram of points after deletion of one site
in time linear to the number of Voronoi neighbors of the deleted site;  
(2) computing the order-$(k{+}1)$ subdivision within an order-$k$
Voronoi region;
(3) computing the farthest Voronoi diagram of point-sites in linear
time, after computing their convex hull.
A much simpler randomized approach for the same problems 
was introduced by Chew~\cite{Chew90}.
The medial axis of a simple polygon is another well-known problem
to admit a linear-time construction, as shown by
Chin\etal\cite{snoyeink99}.

Surprisingly, no linear-time constructions have been known for any of the problems~(1)-(3)
for 
Voronoi diagrams concerning non-point sites,  and
similarly for abstract Voronoi diagrams.
Under restrictions, Klein and Lingas~\cite{KL94} adapted the
linear-time approach of~\cite{AGSS89}
to the abstract framework,
showing that a \emph{Hamiltonian
	abstract Voronoi diagram} can be computed in linear time, given the
order of Voronoi regions along an unbounded simple curve, which visits
each region \emph{exactly once} and can intersect each bisector 
only once.  This construction has been extended recently to include
some forest
structures~\cite{BKLL18}, under similar restrictions, where no region 
can have multiple faces within a domain enclosed by a curve, and each
bisector can intersect this domain in one component. 

In this paper 
we consider the fundamental problem of site-deletion in abstract
Voronoi diagrams and provide a simple expected
linear-time technique to achieve this task.
We work in  the framework of abstract
Voronoi diagrams so that we can simultaneously address all the concrete
instances that fall under their umbrella.
After deletion,
we extend the randomized linear-time technique to the remaining problems:
(cfr.~2) computing the order-$(k{+}1)$ subdivision within an order-$k$
abstract Voronoi region; and 
(cfr.~3) computing the farthest abstract Voronoi diagram, after 
the order of its faces at infinity is known.
To the best of our knowledge, no deterministic linear-time technique
is yet known for these problems.
In the process, we define a \emph{Voronoi-like diagram}, a relaxed
Voronoi structure, which is interesting in its own right.
\emph{Voronoi-like regions} are supersets of real Voronoi regions, and
their boundaries correspond to simple \emph{monotone paths} in the
arrangement of the underlying bisector system (see Definition~\ref{def:monotone}).
We prove correctness 
and uniqueness of this structure and use it to
derive a very simple technique to address the above problems in
expected linear time.

An earlier attempt towards a linear-time construction for the
farthest-segment Voronoi diagram appeared in \cite{KP17arxiv},
following a different geometric formulation for segments, which however, 
does not extend to the abstract setting.
A preliminary version of this paper regarding site deletion in
abstract Voronoi diagrams appeared in \cite{JP18}.
In three dimensions, 
site-deletion in Delaunay triangulations of point-sites, as inspired by the 
randomized approach of Chew~\cite{Chew90}, has been considered in~\cite{BDMSS13}.

\begin{figure}
	\begin{minipage}{0.48\textwidth}
		\centering
		\includegraphics{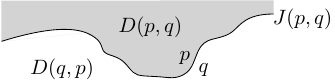}
                \vspace*{0.5\baselineskip}
		\caption{
			A bisector $J(p,q)$ and its two dominance regions; 
			$D(p,q)$ is shown shaded.
                      }
		\label{fig:bisector}
	\end{minipage}
	\hfill
	\begin{minipage}{0.48\textwidth}
		\centering
		\includegraphics[page=2]{Figures/3bisectors_region}
                \caption{
                  The Voronoi diagram of 3 sites, the underlying
                  bisector system in dashed lines, and 
                  $\VR(p,\{p,q,r\})$ shaded.
				}
		\label{fig:3-bis}
	\end{minipage}
\end{figure}
\paragraph{Abstract Voronoi diagrams (AVDs).}
These diagrams 
were introduced by Klein~\cite{K89}.
Instead of sites and distance measures, they 
are defined in terms of 
bisecting curves 
that satisfy some simple combinatorial properties.
Given a set $S$  of $n$  abstract sites, 
the bisector $J(p,q)$  of two sites $p,q \in S$ is an unbounded 
Jordan curve, homeomorphic to a line, that divides the plane into 
two open domains: the \emph{dominance region of $p$}, $D(p,q)$ 
(having label $p$),
and the \emph{dominance region of $q$}, $D(q,p)$ 
(having label $q$), see Fig.~\ref{fig:bisector}. 
The \emph{Voronoi region} of $p$ is
\[\VR(p,S) = \bigcap_{q \in S \setminus \{p\}}  D(p,q).\]
The (\emph{nearest-neighbor}) \emph{Voronoi diagram} of $S$ is 
\[\V(S) = \mathbb{R}^2\setminus \bigcup_{p \in S}\VR(p, S).\]

Following the traditional model of AVDs (see
e.g.~\cite{K89,BCKLPZ15,BKLL18})  
the bisector system is assumed to satisfy the following axioms, 
for every subset $S' \subseteq S$:

\begin{enumerate}
\itemsep=0pt
	\item[(A1)] Each Voronoi region $\VR(p, S')$ is non-empty and path-connected.
	\item[(A2)] Each point in the plane belongs to the closure of a Voronoi region $\VR(p, S')$.
	\item[(A3)]  
          Each bisector $J(p,q)$ is an unbounded curve, which after
          stereographic projection to the sphere can be completed
          to a closed Jordan curve through the north
          pole. 
	\item[(A4)] Any two bisectors $J(p, q)$ and $J(r,t)$ intersect
	transversally and in a finite number of points. (It is possible to relax this axiom, see \cite{KLN09}).
\end{enumerate}

The abstract Voronoi diagram $\V(S)$ is a plane graph of structural
complexity $O(n)$ whose regions are 
simply connected. 
It can be computed in time $O(n\log n)$, both randomized
\cite{KMM93} and  deterministic~\cite{K89}.

To update $\V(S)$ after deleting one site $s\in S$, we need to compute
$\V(S\setminus \{s\})$ within  $\VR(s,S)$.
This diagram 
is a tree, if $\VR(s,S)$ is
bounded, and a forest otherwise. However, its regions can be
disconnected, and one region may consist of multiple faces.
In fact, site-occurrences
along  $\partial \VR(s,S)$ form a 
Davenport-Schinzel sequence of order~2.
Disconnected regions introduce severe  complications 
and constitute  a major difficulty, which differentiates the problem
from its counterpart on point sites.
For example, let $S'\subset S\setminus \{s\}$;
the
diagram 
$\V(S') \cap \VR(s, S' \cup \{s\})$ may contain faces that do not
even appear in 
$\V(S\setminus
\{s\})\cap\VR(s,S)$, and conversely, an arbitrary sub-sequence
of arcs on
$\partial\VR(s,S)$ need not be related to any Voronoi diagram.
At a first sight, a linear-time algorithm may even seem infeasible. 

\noindent

\paragraph{Our results.}
In this paper we formalize the concept of a \emph{Voronoi-like diagram}, 
a relaxed Voronoi structure defined as a
graph (a tree or forest) in the arrangement of the underlying bisector system,
and prove that it is well-defined and unique. 
This structure provides 
the tool we need to deal with
disconnected Voronoi regions, and thus, address the site-deletion
problem efficiently.
Given a Voronoi-like diagram, we define an \emph{insertion operation} and prove its
correctness. 
This makes a simple randomized incremental construction possible.
The time-analysis of the randomized algorithm  is non-standard
because the intermediate 
Voronoi-like structures are order-dependent.
We give a technique, which partitions the permutations of length $i$ into manageable
groups of $i$ permutations each, and show that the time
complexity per group is $O(i)$, 
deriving that each insertion
step can be performed in expected  $O(1)$ time.
This technique  may be independently useful in deriving expectation
in order-dependent cases.
In this paper we focus on site-deletion, computing 
$\V(S\setminus \{s\})\cap\VR(s,S)$
in expected time $O(|\partial\VR(s,S)|)$, i.e., in expected time linear 
in the number of Voronoi neighbors of the deleted site.
We also extend  the approach to address problems~(2) and~(3) for the
order-$k$ and the farthest abstract Voronoi diagram  respectively.
The sequence of the faces 
at infinity of the latter diagram can be computed in time $O(n\log n)$. 
\deleted{
The farthest abstract Voronoi diagram can be computed in expected
linear time, after the sequence of its faces at 
infinity is known (see Section~\ref{sec:farthest}). 
The latter sequence can be computed in time $O(n\log n)$.
}

Examples of concrete diagrams that fall under the AVD
umbrella and thus they benefit from our approach 
include: 
disjoint line segments and disjoint convex polygons of constant size
in the $L_p$ norms, or under the Hausdorff metric;
point-sites in any convex distance metric or the Karlsruhe metric;
additively weighted points that have non-enclosing circles; 
power diagrams with non-enclosing circles.

\section{Preliminaries} \label{sec:defs}

Let $S$ be a set of $n$ abstract \emph{sites} (a set of indices) that define an
\emph{admissible} system of bisectors 
$\J= \{J(p, q) : p \neq q \in S\}$.
$\J$ fulfills axioms (A1)--(A4) for every 
$S'\subseteq S$.

Bisectors in $\J$ that have a site  $p$ in common are called \emph{$p$-related}
or simply \emph{related}.
Any two related bisectors can intersect at most twice~\cite[Lemma
3.5.2.5]{K89}.
When two related bisectors $J(p,q)$ and $J(p,r)$ intersect, 
bisector $J(q,r)$  also intersects with them at the same 
point(s), and these points are the Voronoi vertices of the diagram
$\V(\{p,q,r\})$.
The Voronoi diagram of three sites $\V(\{p,q,r\})$ may have one or two
(or none)
Voronoi vertices, see Fig.~\ref{fig:3-bis}.
The set of all $p$-related bisectors that involve  sites in
$S'\subseteq S$ is denoted 
$\jj{p}{S'} = \{J(p,q) \,|\, q \in S', q \neq p\}$.

Let $\VR(s,S)$ be the Voronoi region of a site $s\in S$.
Although  $\VR(s,S)$ is simply connected, the sites in  $S\setminus
\{s\}$ that appear 
along  the  boundary  $\partial \VR(s,S)$  may repeat, forming a
Davenport-Schinzel sequence of order~2.
This is because $s$-related bisectors can intersect at most
twice, and thus, \cite[Theorem
5.7]{sharir95} applies. 
This is a fundamental difference from the classic case of point-sites in the Euclidean plane, where 
bisectors are straight-lines, therefore they intersect once, and no site repetition can occur along
$\partial \VR(s,S)$.

Suppose we delete the site $s\in S$ from $\V(S)$. To update the Voronoi diagram
after the deletion of $s$,
 we need to compute
$\V(S\setminus \{s\})$ within the Voronoi region $\VR(s,S)$, i.e.,
compute $\V(S\setminus \{s\}) \cap \VR(s,S)$.
We first characterize the structure of this diagram in the following lemma.
An alternative proof and characterization can be
derived from the order-$k$ counterpart~\cite{BKL19}, however, this
proof appeared later, after the preliminary version of
this paper~\cite{JP18}.

	\begin{figure}
	\centering
	\begin{minipage}{.48\textwidth}
		\centering
		\includegraphics{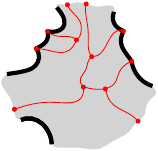}
		\caption{$\V(\sms) \cap \VR(s,S)$ in
			red; $\partial \VR(s,S)$ is shown in
			bold black.
		}
		\label{fig:unbounded}
	\end{minipage}
	\hfill
	\begin{minipage}{.48\textwidth}
		\includegraphics{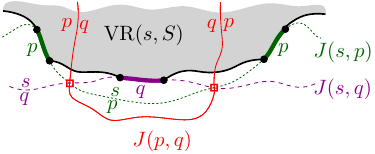}
		\caption{$\VR(p,\sms) \cap \VR(s,S)$ cannot be connected because of $J(p,q)$. 
		}
		\label{fig:treelike}
	\end{minipage}
\end{figure}

\begin{lemma}\label{lem:treelike}
	$\V(S\setminus \{s\})\cap \VR(s,S)$ is a forest, having exactly one
	face for each Voronoi edge of $\partial \VR(s,S)$. Its leaves are the Voronoi vertices of
	$\partial \VR(s,S)$, and points at infinity if $\VR(s,S)$ is
	unbounded (see Fig.~\ref{fig:unbounded}). If $\VR(s,S)$ is bounded then $\V(S\setminus \{s\})\cap
	\VR(s,S)$ is a tree.
\end{lemma}
\begin{proof}

	Every face in $\V(S\setminus \{s\})\cap \VR(s,S)$ must 
	touch the boundary $\partial \VR(s,S)$ 
	because Voronoi regions are non-empty and connected; 
	this implies that the diagram is a forest.
	Every Voronoi edge $e \subseteq J(s,p)$ on $\partial \VR(s,S)$
	must be
	entirely in $\VR(p,\sms)$.  
	Thus, no leaf can lie in the interior of a Voronoi edge of
	$\partial \VR(s,S)$.
	On the other hand, 
	each Voronoi vertex of $\partial \VR(s,S)$ must be 
	a leaf of the diagram as its incident  
	edges are induced by different sites.
	
	Now we show that no two edges of $\partial \VR(s,S)$ can be
	incident  to the same face of $\V(\sms) \cap \VR(s,S)$.
	Consider two edges on $\partial \VR(s,S)$ induced by the same  site $p
	\in \sms$. 
	Then there exists an edge between them, induced by a site $q \neq p$,
	such that the bisector $J(s,q)$ has exactly two intersections with $J(p,s)$
	as shown in Fig.~\ref{fig:treelike}.
	The bisector $J(p,q)$ intersects with them at the same two
	points.
	Since the bisector system is admissible, and thus 
	$\VR(p,\{s,p,q\})$ is connected,  $J(p,q)$ connects these endpoints
	through $D(p,s)\cap D(q,s)$ as shown in Fig.~\ref{fig:treelike},
	thus, $J(p,q) \cap \VR(s, \{s,p,q\})$ consists of two unbounded
	connected  components. 
	This implies that $D(p,q) \cap \VR(s,S)$ must have two disjoint faces,
	each of which is incident to exactly one of the two edges of $p$. 
	Thus, $\VR(p,\sms) \cap \VR(s,S)$ cannot be connected and 
	the two edges of $p$ must be incident to different faces of
	$\V(S\setminus \{s\})\cap \VR(s,S)$. 
	
	If  $\VR(s,S)$ is unbounded, two consecutive edges of 
	$\partial \VR(s,S)$ can extend to infinity, in which 
	case there is at least one edge of $\V(S\setminus \{s\})\cap \VR(s,S)$ extending 
	to infinity between them; thus, leaves can be points at infinity. 
	If $\VR(s,S)$ is bounded,
	all leaves of  $\V(S\setminus \{s\})\cap \VR(s,S)$ must lie
	on $\partial \VR(s,S)$. Since no face 
	is incident to more than one edge of  $\partial \VR(s,S)$, 
	in this case   $\V(S\setminus \{s\})\cap \VR(s,S)$ 
	cannot be disconnected, and thus is a tree. 
\end{proof}

Let $\Gamma$ be a closed Jordan curve in the plane large enough to
enclose all the 
intersections of bisectors in $\J$, and such that each bisector intersects $\Gamma$ exactly 
twice and transversally.
To avoid dealing with infinity, and without any loss of generality, we restrict all
computations within~$\Gamma$.\footnote{The presence of $\Gamma$ is
	conceptual and its exact position unknown; 
	we never compute coordinates on $\Gamma$.}
The curve $\Gamma$ can be interpreted 
as $J(p,s_\infty)$, for all $p \in S$, where $s_\infty$ is an
additional site at infinity.
Let $D_\Gamma$  denote the portion of the plane  enclosed by
$\Gamma$.
$\VR(s,S) \cap D_\Gamma$ is the \emph{domain} of our computation.
Fig.~\ref{fig:boundaryNVD} illustrates possible cases of the
computation domain.

\begin{figure}
	\centering
	\includegraphics{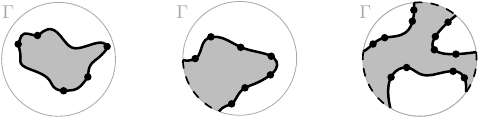}
	\caption{
		The domain of computation $\VR(s,S) \cap D_\Gamma$ (shaded).}
	\label{fig:boundaryNVD}
\end{figure}

We first make some observations regarding an admissible
bisector system, which we use as tools in the proofs throughout this paper.
Let $C_p$ be a cycle of $p$-related bisectors in the arrangement
of bisectors $\J\cup \Gamma$.
If the label $p$ appears on the outside of the cycle for every edge in $C_p$,
then $C_p$ is called \emph{$p$-inverse}, see Fig.~\ref{fig:cycles}(a).
If the label $p$ appears only inside $C_p$ then $C_p$ is
called a \emph{$p$-cycle}, see Fig.~\ref{fig:cycles}(b).
Recall that $\Gamma$ can be considered a $p$-related bisector, for all sites $p\in
S$, where the label $p$ is in the interior of $\Gamma$.
Thus, a $p$-cycle may contain pieces of $\Gamma$, whereas a $p$-inverse
cycle cannot contain any such piece.

\begin{figure}
	\centering
	\includegraphics{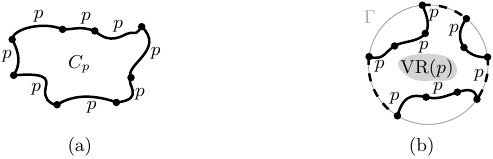}
	\caption{(a) A $p$-inverse cycle. (b) A $p$-cycle.}
	\label{fig:cycles}
\end{figure}

\begin{lemma}\label{lem:nocycle}
	In an admissible bisector system there is no $p$-inverse cycle.
\end{lemma}

\begin{proof}
  Suppose a $p$-inverse cycle exists in the admissible bisector
  system. Let $C_p$ denote a minimal such cycle, where  no $p$-related
  bisector may intersect the interior of  $C_p$ and let $D_p$ denote the
  interior of $C_p$. 
  Such a minimal cycle must exist because if a bisector $J(p,q)$ intersects $D_p$, 
 then it defines another (smaller) $p$-inverse cycle  that is
 contained in $C_p\cup D_p$
   and whose interior is not intersected by
 $J(p,q)$.
Let $S'\subseteq S$ denote  the set of sites that define 
the	edges of $C_p$. 
Considering $S'$, the 
farthest Voronoi region of $p$ is $\freg(p,S') = \bigcap_{q \in S' \setminus \{p\}}
D(q,p)$.
But by its definition, $D_p$ must be identical to one face of $\freg(p,S')$.
Since farthest Voronoi regions must be unbounded 
\cite{AbstractFarthestVoronoi,BCKLPZ15}, 
we derive a contradiction.   
\end{proof}

The following \emph{transitivity lemma} is a consequence of transitivity of
dominance regions~\cite[Lemma 2]{BCKLPZ15} and the fact that bisectors $J(p,q),J(q,r)$,
$J(p,r)$ intersect at the same point(s).
Let  $\overline{X}$ denote the closure of a region $X$.

\begin{lemma} \label{lem:transitivity}
	Let $z \in \RR$ and $p,q,r \in S$.
	If $z \in D(p,q)$ and $z \in \overline{D(q,r)}$, then $z \in D(p,r)$.
\end{lemma}

We  make a general position assumption that no three  $p$-related bisectors
intersect at the same point.
This implies that Voronoi vertices have  degree 3.

\section{Problem formulation, definitions and properties}\label{sec:problem}

Let $\arcs$ denote the sequence of Voronoi edges bounding the Voronoi
region $\VR(s,S)$ within the domain $D_\Gamma$,
i.e., 
$\arcs = \partial \VR(s,S) \cap D_\Gamma$. 
We consider $\arcs$ as a  cyclically ordered set of \emph{arcs},
where each arc is a portion of an $s$-related bisector defining a
Voronoi edge along  $\partial \VR(s,S)$.
A single site in $S\setminus\{s\}$  may induce several arcs in
$\arcs$.
For any arc $\alpha\in \arcs$, let $s_\alpha$ 
denote the site in $S$ such that
$\alpha \subseteq J(s,s_\alpha)$.

We can interpret the arcs in $\arcs$ as sites that induce a
Voronoi diagram $\V(\arcs)$ such that
$\V(\arcs)
= \V(\sms)\cap \VR(s,S) \cap D_\Gamma$, see
Fig.~\ref{fig:vd}.
By Lemma~\ref{lem:treelike}, each face of $\V(\arcs)$ is incident
to exactly one arc in $\arcs$.
Thus, the face of $\V(\arcs)$ incident to an arc $\alpha\in\arcs$ can
be considered its Voronoi region $\fs{\alpha}$.
Then $\V(\arcs)$ can be regarded as the diagram derived by the 
Voronoi regions of the arcs in~$\arcs$.

\begin{figure}
	\centering
	\includegraphics{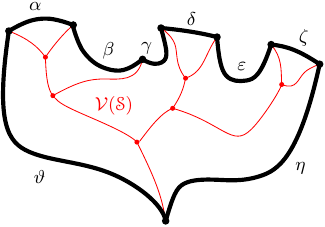}
	\caption{Illustration of $\arcs=\partial \VR(s,S)$ in bold
		(black) and $\V(\arcs)$ in red; 
		$\arcs  = ( \alpha, \beta, \gamma,\delta, \varepsilon, \zeta, \eta, \vartheta )$. 
	}
	\label{fig:vd}
\end{figure}

        The arrangement of a bisector set $\J'\subseteq \J$ is denoted by $\arr(\J')$.
	A \emph{path}~$P$ in the arrangement $\arr(\J')$ is a connected sequence 
	of alternating edges and vertices in this arrangement. 
	An \emph{arc} $\alpha$ of $P$  (denoted as $\alpha \in P$) is a maximally connected
        collection 
	of consecutive edges and vertices of the arrangement along $P$,  
	which belong to the same bisector.
	The common endpoint of two consecutive arcs of $P$ is a
	\emph{vertex of $P$}. An arc of $P$ is also called an \emph{edge}.
Two consecutive arcs in a path $P$  are pieces of 
different bisectors.


\begin{figure}[b]
		\centering
		\includegraphics{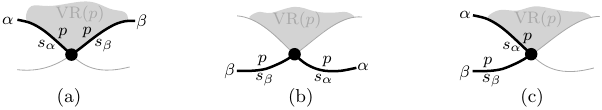}
		\caption{(a) Arcs $\alpha, \beta$ 
			fulfill the $p$-monotone path condition;  
			they do not fulfill it (b) and~(c).}
		\label{fig:pathvertex}
\end{figure}

\begin{definition} \label{def:monotone}
	A path in the arrangement of $p$-related bisectors $\jj{p}{S'}$, $S'\subseteq S$,  is called \emph{$p$-monotone} 
	if any two consecutive arcs $\alpha, \beta$ along this path,  where
	$\alpha\subseteq J(p,s_\alpha)$ and $\beta\subseteq J(p,s_\beta)$, 
	coincide locally (within a neighborhood of their common
        endpoint) with  Voronoi edges  of  
	$\partial \VR(p,\{p,s_\alpha,s_\beta\})$ that are incident to
	this common endpoint 
	(see Fig.~\ref{fig:pathvertex} and Fig.~\ref{fig:paths}).
\end{definition}

\begin{definition} \label{def:envelope} 
         The  \emph{$p$-envelope} (or simply \emph{envelope}) of $\jj{p}{S'}$
                is $\env(\jj{p}{S'})=\partial \VR(p,S'  \cup \{p\})$ (see
        Fig.~\ref{fig:paths}(a)).
\end{definition}

The arrangement of the  bisectors in $\J_{p,S'}$ may  consist of several connected
components.
We can unify these  connected components by including $\Gamma$ in the
bisector system.
Then, $\env(\jj{p}{S'}\cup  \Gamma)$ is a single closed $p$-monotone
path, which 
contains all the connected components of $\env(\jj{p}{S'})$
interleaved by arcs of $\Gamma$.
$\VR(s,S)$ and in particular to a subset $\arcs'$ of its Voronoi edges.

\begin{figure}
	\centering
	\includegraphics{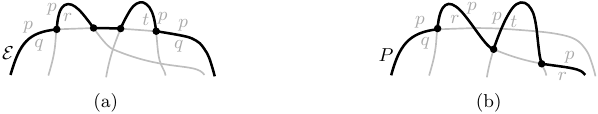}
	\caption{
		(a) The envelope $\E=\env(\mathcal{J}_{p,\{q,r,t\}})$.
		(b) A $p$-monotone path $P$ in  $\mathcal{J}_{p,\{q,r,t\}}$.}
	\label{fig:paths}
\end{figure}

\begin{definition} \label{def:boundarycurve}
	Consider $\arcs'\subseteq \arcs$ and let 
	$S'=\{s_{\alpha }  \in S \, | \, \alpha\in \arcs'\}$ 
	be the corresponding set of its sites.
        A \emph{boundary curve} $\pp$ for $\arcs'$ is a
	closed $s$-monotone path in the arrangement of $s$-related
        bisectors $\J_{s,S'} \cup \Gamma$ such that all arcs in
        $\arcs'$ are contained in $\pp$.
        Let $D_{\pp}$ denote the \emph{domain} of $\pp$, which is 
	the part of the	plane enclosed by $\pp$. 
        Let  $S_\pp=S'$. 
      \end{definition}

A set $\arcs'\subset\arcs$ can admit several different boundary curves;
one such boundary curve is  its
 $s$-envelope $\E=env(\arcs')=\env(\J_{s,S'}\cup \Gamma)$.
The set $\arcs$ can admit only one boundary curve, which is 
its $s$-envelope $env(\arcs)=\partial(\VR(s,S) \cap D_\Gamma)$.
Fig.~\ref{fig:vld} illustrates  a boundary curve for a subset of
arcs from Fig.~\ref{fig:vd}.

A boundary curve $\pp$  
consists of pieces of $s$-bisectors 
called  \emph{boundary arcs}, and pieces of $\Gamma$,  
called \emph{$\Gamma$-arcs}.
$\Gamma$-arcs correspond to openings of the 
domain $D_\pp$ to infinity.
Among the boundary arcs, those containing an arc of $\arcs'$ are called
\emph{original} and  others are called \emph{auxiliary arcs}.
Original boundary arcs in~$\pp$ are expanded versions of the arcs in~$\arcs$.
To distinguish  them, we call the elements of $\arcs$
\emph{core arcs} and use an~$^*$ in their notation. 
We denote by $|\pp|$ the number of boundary arcs in~$\pp$.
Fig.~\ref{fig:vld} illustrates a boundary curve $\pp$ on $\arcs' \subseteq \arcs$ 
consisting of five original arcs, one auxiliary arc ($\beta'$)
and one $\Gamma$-arc ($g$); the core arcs 
are illustrated in bold; the set $\arcs$ is shown in  Fig.~\ref{fig:vd}.

\begin{figure}[b]
	\centering
	\includegraphics{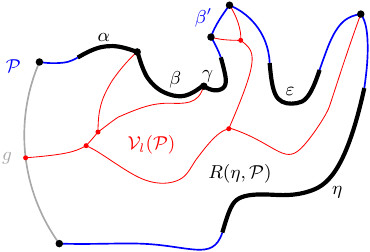}
	\caption{
		A boundary curve $\pp$ on $\arcs'\subseteq
		\arcs$, where $\arcs'$ is shown in bold,  and
                its Voronoi-like diagram $\vld(\pp)$ shown in thin red.
		The gray arc  $g$ is a $\Gamma$-arc, and the blue arc $\beta'$ is
		an auxiliary arc; the remaining arcs are original.
		The set $\arcs$ is shown in Fig.~\ref{fig:vd}.
	}
	\label{fig:vld}
      \end{figure}

We now define the Voronoi-like diagram of a boundary curve $\pp$ on
$\arcs'\subseteq \arcs$, where
$S'=\{s_{\alpha } \in S \, | \, \alpha\in \arcs'\}$
is the corresponding set of sites.
Let $\J(S') \subseteq \J$ be the system of bisectors related to
$S'$, i.e.,
$\J(S') = \{J(p,q) \in \J \,|\, p,q \in S'\}$.
    
\begin{definition} \label{def:vld} 
  	Given a boundary curve $\pp$ 
  	on  $\arcs'\subseteq\arcs$, 
    the \emph{Voronoi-like diagram} of  $\pp$ is a plane graph
    $\vld(\pp)$ on 
    the arrangement  of the bisector system  $\J(S') $
    inducing a subdivision of the domain $D_{\pp}$ as follows 
	(see Fig.~\ref{fig:vld}):
        \begin{enumerate}
        \itemsep=0pt
	\item
	for each boundary arc $\alpha\in \pp$, there is exactly one distinct face
	$R(\alpha,\pp)$, whose boundary is
	an $s_\alpha$-monotone path in $\jj{s_\alpha}{S'} \cup
	\Gamma$, plus arc $\alpha$;	
	\item the faces cover the domain $D_{\pp}$: $\bigcup_{\alpha \in
            {\pp\setminus\Gamma}} \overline{R(\alpha,\pp)} =
          \overline{D_\pp}$.
\end{enumerate}
 \end{definition}

Voronoi-like regions in $\vld(\pp)$ are related to the real Voronoi
regions  
as supersets as we show in the following lemma.
Let $\V(\E)=\V(S')\cap D_{\E}$ be the Voronoi diagram of the $s$-envelope $\E$ of
$\arcs'$.
In $\V(\E)$ any face incident to a boundary
arc $\alpha \in \E$ can be regarded as its  Voronoi region 
$\VR(\alpha,\E)$.
For an original arc $\alpha\in \pp$ there is an original arc
$\tilde\alpha\in \E$ and a core arc $\alpha^*\in \arcs$ such that
$\alpha\supseteq \tilde\alpha\supseteq \alpha^*$.

\begin{lemma} \label{lem:facesupset}
  Let $\alpha\in \pp$ be a boundary arc such that $\tilde\alpha\subseteq \alpha$
   appears  on  the $s$-envelope $\E$. 
   Then, $R(\alpha,\pp) \supseteq \VR(\tilde\alpha,\E)$.
	Further, if $\alpha$ is original, then $
	R(\alpha,\pp) \supseteq \VR(\tilde\alpha,\E)\supseteq\VR(\alpha^*,\arcs)$.
      \end{lemma}

\begin{proof}
	By the definition of a Voronoi  
	region, no piece of a bisector $J(s_\alpha,\cdot)$  
	can appear in the interior of a Voronoi region in
        $\V(S')\cap D_\E$. Thus  no piece of
        $J(s_\alpha,\cdot)$ can appear in  
        $\VR(\tilde\alpha,\E)$, for any 
	$\tilde\alpha\in \E$. 
        Since $\alpha \supseteq\tilde\alpha$, 
        by the definition of a
        Voronoi-like region  it follows that  $R(\alpha,\pp) \supseteq
        \VR(\tilde\alpha,\E)$.  
	For an original arc $\alpha$, since  $S' \subseteq S$, by the
	monotonicity property of Voronoi regions, we also have 
	$\VR(\tilde\alpha,\E)\supseteq\VR(\alpha^*,\arcs)$.
\end{proof}

In Fig.~\ref{fig:vld} the Voronoi-like region $R(\eta,\pp)$ 
is a superset of its corresponding Voronoi region $\VR(\eta^*,\arcs)$ 
of $\V(\arcs)$ in Fig.~\ref{fig:vd};
similarly  $R(\alpha,\pp)\supseteq \VR(\alpha^*,\arcs)$.

As a corollary to the superset property of Lemma~\ref{lem:facesupset}, the adjacencies
of the real Voronoi diagram $\V(\E)$ are preserved in
$\vld(\pp)$,  for all the original arcs.
As a result, $\vld(\E)$ must coincide with the real  Voronoi diagram
$\V(\E)=\V(S')\cap D_\E$.

\begin{corollary}\label{lem:EVLD} 
	$\vld(\E)=\V(\E)=\V(S')\cap D_\E$ for the $s$-envelope $\E$ of
        $\arcs'\subseteq \arcs$.
\end{corollary}

\begin{proof}
Consider two arcs $\alpha \neq \beta$ of $\E$. 
Suppose $\VR(\alpha,\E)$ is adjacent to $\VR(\beta,\E)$.
Since by Lemma~\ref{lem:facesupset}, 
$R(\alpha,\E) \supseteq \VR(\alpha,\E)$ and 
$R(\beta,\E) \supseteq \VR(\beta,\E)$, it follows that 
$R(\alpha,\E)$ must be adjacent to $R(\beta,\E)$.
This implies that the regions in $\vld(\E)$ have the  
same adjacencies as in $\V(\E)$. Lemma~\ref{lem:facesupset} also
implies that 
there can be no additional adjacencies;
thus, $\vld(\E)=\V(\E)$.
 \end{proof}

In the remaining section we give basic properties of Voronoi-like
regions involving their interaction with the bisectors in $\J$, which we use 
 to derive correctness and establish that the Voronoi-like
diagram is well-defined.

\subsection{Properties of Voronoi-like regions}

The following property establishes that  a Voronoi-like
region $R(\alpha,\pp)$ can never be intersected by  $J(s,s_\alpha)$.

\begin{lemma}\label{lem:regionclosertos}
 	For any arc $\alpha\in \pp$, $R(\alpha,\pp) \subseteq D(s,s_\alpha)$.
\end{lemma}
\begin{proof}
	The contrary would yield a forbidden $s_\alpha$-inverse 
	cycle defined by a component of $J(s,s_\alpha)\cap
        R(\alpha,\pp)$ and the incident  portion of $\partial R(\alpha,\pp)$.
 \end{proof}

\begin{lemma}\label{lem:nocycleindomain}
  For a boundary curve $\pp$, its domain
	$\overline{D_\pp}$ may not contain a $p$-cycle formed by the
        bisectors of
        $\J(S_\pp)\cup\Gamma$,
        for any site $p\in S_\pp$.
\end{lemma}

\begin{proof}         
  Let $p\in S_\pp$. 
  Any original arc of $p$ in $\pp$ is
  	bounding $\VR(p, S_\pp\cup\{s\})$, thus, it must have a portion
  	within the interior of $\VR(p, S_\pp)$ in  $\V(S_\pp)$.
        Hence, $\VR(p,S_\pp)$ must have some non-empty portion outside 
        the closure of $D_\pp$. 
	However, $\VR(p,S_\pp)\cap D_\Gamma$ must be enclosed within any  
	$p$-cycle of $\J(S_\pp)\cup\Gamma$, by its definition.
	Thus, no such $p$-cycle can be contained in $\overline{D_\pp}$.
 \end{proof}

Next, we give a key property of a Voronoi-like region, 
which we call
the \emph{cut property}, see Fig.~\ref{fig:cutproperty}.
Suppose
bisector $J(s_\alpha,s_\beta)$ intersects the Voronoi-like region
$R(\alpha,\pp)$. 
Let $e$ be a connected  component of $J(s_\alpha,s_\beta)\cap
R(\alpha,\pp)$ and let 
 $R_e(\alpha)$ denote the portion of  region $R(\alpha,\pp)$
 that is \emph{cut out by $e$} as shown in
 Fig.~\ref{fig:cutproperty}.
More precisely  $R_e(\alpha)$ is defined as follows.
If $e$ does not intersect $\alpha$, 
then $R_e(\alpha)$ is the portion of the region
at the opposite side of $e$ as $\alpha$ (case (a),  see Fig.~\ref{fig:cutproperty}(a)). 
Otherwise, let $\tilde\beta$ be the component of $J(s,s_\beta)\cap
R(\alpha,\pp)$  incident to $e$ and $\alpha$, and let $R_e(\alpha)$ be
the portion of  $R(\alpha,\pp)$ that contains $\tilde\beta$ (cases (b) and
(d) in Fig.~\ref{fig:cutproperty}). 
If there is another component  of $J(s_\alpha,s_\beta)\cap
R(\alpha,\pp)$ incident to
$\alpha$, let $R_e(\alpha)$ be the portion of $R(\alpha,\pp)$  
between  the two components (case~(c), see  Fig.~\ref{fig:cutproperty}(c)).
\deleted{
If $e$ and $\tilde\beta$  have only one endpoint on $\alpha$, let
$R_e(\alpha)$ be the portion of the region that contains $\tilde\beta$
(case~(b), see Fig.~\ref{fig:cutproperty}(b)). 
If $e$ intersects $\alpha$ once but $\tilde\beta$ intersects it twice,
then there are 
two components of $J(s_\alpha,s_\beta)\cap R(\alpha,\pp)$ incident
to $\alpha$;
let $R_e(\alpha)$ denote the portion of $R(\alpha,\pp)$  
between these two components (case~(c), see Fig.~\ref{fig:cutproperty}(c)).
Otherwise, both $e$ and $\tilde \beta$ intersect $\alpha$ twice, 
let $R_e(\alpha)$  be the
portion of the region incident to  $\partial R(\alpha,\pp)$
(case~(d), see Fig.~\ref{fig:cutproperty}(d)).
}
Note that  if $\beta \in \pp$ then cases (c) and (d) cannot appear
since related bisectors can only intersect twice.

\begin{figure}
	\centering
	\includegraphics{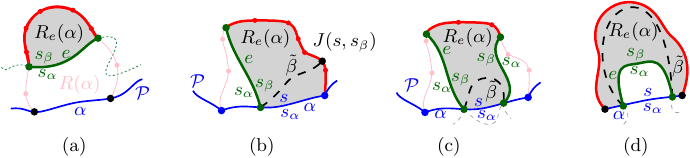}
	\caption{Illustrations for Lemma~\ref{lem:propertyR}.
		The  shaded region $R_e(\alpha)$ lies in $D(s_\beta,s_\alpha)$.        
	}
	\label{fig:cutproperty}
\end{figure}

\begin{lemma}\label{lem:propertyR}
	Suppose bisector $J(s_\alpha,s_\beta)$ appears within
	$R(\alpha,\pp)$ (see Fig.~\ref{fig:cutproperty}).
	For any connected component $e$ of $J(s_\alpha,s_\beta)\cap
	R(\alpha,\pp)$, it holds
        $R_e(\alpha) \subseteq
	D(s_\beta,s_\alpha)$. 
        Thus, if $e$ does not intersect $\alpha$, the label
	$s_\alpha$ must appear on the same side of $e$ as $\alpha$.
\end{lemma}

\noindent
Note that $\partial R_e(\alpha)$ may  contain  $\Gamma$-arcs.

\begin{proof}
Let  $e$ be an arbitrary component  of $J(s_\alpha,s_\beta)\cap
R(\alpha,\pp)$.        
Suppose for the sake of contradiction that $R_e(\alpha) \not\subseteq
        D(s_\beta,s_\alpha)$. 
Then $J(s_\beta,s_\alpha)$  must intersect the interior of
$R_e(\alpha)$ with  a component $e'$ of
$J(s_\beta,s_\alpha) \cap R(\alpha,\pp)$,  which is different from
$e$. 
Among any such component, let $e'$ be the first one following $e$
along $J(s_\beta,s_\alpha)$.
Since  $e'$ cannot intersect $e$, nor  can it intersect $\tilde \beta$,
it follows that $e'$ must create 
an $s_\alpha$-cycle with $\partial R_e(\alpha)$,
contradicting Lemma~\ref{lem:nocycleindomain}.
 \end{proof}

Lemma~\ref{lem:propertyR} implies that any components of 
$J(s_\alpha,s_\beta)\cap R(\alpha,\pp)$ must appear
sequentially along $\partial R(\alpha,\pp)$.
In addition, if any such component exists, $J(s,s_\beta)$ 
must also intersect the domain  $D_\pp$ with a component that is  
 \emph{missing} from $\pp$.
We use this fact to establish that $\vld(\pp)$
is unique in the following theorem whose
 proof is deferred to
Section~\ref{sec:unique}.

\begin{theorem}\label{thm:vldunique}
	Given a boundary curve $\pp$ of $\arcs'\subseteq\arcs$,
	$\vld(\pp)$ is unique, assuming it exists. 
\end{theorem}

The complexity of $\vld(\pp)$ is $O(|\pp|)$ as it is a planar graph with 
exactly one face per boundary arc  and vertices of degree 3 (or 1). 

\section{Insertion in a Voronoi-like diagram}\label{sec:insert}

Consider a boundary curve $\pp$ on a set of core arcs $\subs  \subset \arcs$ 
and its Voronoi-like diagram $\vld(\pp)$. 
Let $\beta^*$ be a core  arc in $\arcs\setminus \subs$.
Since $\beta^*$ is a core arc, it must be entirely contained in
the domain  $D_\pp$.
We define an insertion operation $\oplus$, which inserts a core arc $\beta^*$  
in $\pp$, and derives the boundary curve $\pp_\beta=
\pp\oplus\beta^*$ 
and $\vld(\pp_\beta)=\vld(\pp) \oplus\beta^*$.

Given $\pp$ and $\beta^*$, let the original arc 
$\beta\supseteq \beta^*$ be the connected component 
of $J(s,s_\beta)\cap \overline{D_{\pp}}$ that contains $\beta^*$, see Fig.~\ref{fig:updateboundary}.
$\pp_\beta$ is the boundary curve derived from $\pp$ by
substituting its portion between the endpoints of $\beta$, with
$\beta$ itself.
We say that   $\pp_\beta$ is derived from $\pp$ by \emph{inserting}
the core arc $\beta^*$, or equivalently, by inserting the original arc
$\beta$.

\deleted{
We often use the notation $\beta$ interchangeably referring either to an original
arc on $\pp$ or a  core arc $\beta^*$  in
$\arcs\setminus \subs$. 
}

\begin{figure}
	\centering
	\includegraphics{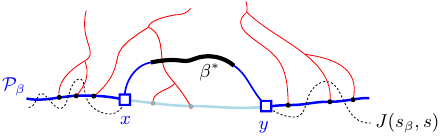}
	\caption{
		$\pp_\beta = \pp \oplus \beta$,
		core arc $\beta^*$ is bold, black.
		Endpoints of $\beta$ are $x,y$. \\ 
	}
	\label{fig:updateboundary}
\end{figure}

The insertion operation $\oplus$ performs the 
following tasks algorithmically:
(1) inserts the core arc $\beta^*$ 
in $\pp$, deriving  $\pp_\beta= \pp\oplus\beta^*=\pp\oplus\beta$;
(2) computes a \emph{merge curve} $J(\beta)$, which defines
  the boundary of $R(\beta, \pp_\beta)$; and
(3) updates $\vld(\pp)$ to derive
    $\vld(\pp_\beta)=\vld(\pp)\oplus \beta$.
Fig.~\ref{fig:insert_arc} enumerates all possible cases in task~1 and 
it is  summarized in the following observation.

\begin{observation} 
	\label{obs:insert-beta} 
	All possible cases of inserting	arc $\beta^*\subseteq \beta$
        in $\pp$ are as follows (see Fig.~\ref{fig:insert_arc}). 
	\begin{enumerate}[(a)]
        \itemsep=-2pt
		\item \label{item:std} 
		Arc $\beta$ straddles the endpoint of two consecutive  
		boundary arcs; no arcs in $\pp$ are deleted.  
		\item \label{item:deletearc} 
		Auxiliary arcs in $\pp$ are deleted by 
		$\beta$;  
		their regions are also deleted from $\vld(\pp_\beta)$.
		\item  \label{item:splitregion}
		An arc $\alpha \in \pp$ is split into two arcs by $\beta$;
		$R(\alpha,\pp)$
		will also be split.
		\item \label{item:splitgap}
		A $\Gamma$-arc  is split in two by $\beta$; $\vld(\pp_\beta)$ may
		switch from being a tree to being a forest.
		\item  \label{item:deletegaps}
		A $\Gamma$-arc is deleted or shrunk by inserting
		$\beta$.  $\vld(\pp_\beta)$ may become a tree.
		\item  \label{trivial}
		$\pp$ already contains a boundary arc $\bar\beta\supseteq \beta^*$; 
		then $\beta=\bar\beta$ and $\pp_\beta=\pp$. 
    \end{enumerate}
               $\pp_\beta$ may contain fewer, the same number, 
               or even one additional auxiliary arc compared to~$\pp$. 
\end{observation}

\begin{figure}
	\centering
	\includegraphics[page=1]{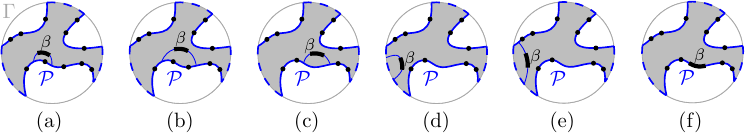}
	\caption{Insertion cases for an arc $\beta$.}
	\label{fig:insert_arc}
\end{figure}

Given $\vld(\pp)$ and arc $\beta$, 
we define  a \emph{merge curve}
$J(\beta)$, which delimits the boundary of 
 $R(\beta, \pp_\beta)$. 
We define $J(\beta)$ algorithmically, starting at an endpoint of
$\beta$, and tracing $s_\beta$-related bisectors within the faces of
$\vld(\pp)$, refer to
Fig.~\ref{fig:updatediagram}.
We  prove that $J(\beta)$ is an $s_\beta$-monotone path that connects the
endpoints of $\beta$.
Let $x,y$ denote the endpoints of $\beta$, where $x\beta y$ appear in
counterclockwise order.
We assume a counterclockwise traversal of $\pp$.

\begin{definition}\label{def:mergecurve}
	Given $\vld(\pp)$ and arc $\beta\subseteq J(s,s_\beta)$, 
	the \emph{merge curve} $J(\beta)$ is a path $(v_1, \dots, v_m)$
	in the arrangement of $s_\beta$-related bisectors, $\jj{s_\beta}{S_\pp} \cup \Gamma$,
	connecting the endpoints of $\beta$,  $v_1 = x$ and $v_m = y$.
 	Each edge $e_i = (v_i, v_{i+1})$ is an arc of a bisector 
	$J(s_\beta, \cdot)$, called a \emph{bisector} edge,  or an arc
        on $\Gamma$.
        We assume a clockwise ordering of $J(\beta)$.
	For $i=1$:   if $x \in J(s_\beta, s_{\alpha})$, then 
	$e_1 \subseteq J(s_\beta,s_{\alpha})$; if $x \in \Gamma$,
	then  $e_1 \subseteq \Gamma$.
	Given $v_i$, vertex $v_{i+1}$ and edge $e_{i+1}$ are defined
	as follows: 

	\begin{enumerate}
        \topsep=0pt
        \itemsep=0pt
		\item
		If   $e_{i} \subseteq J(s_\beta, s_{\alpha})$, 
		let $v_{i+1}$ be the other endpoint of the 
		component  $J(s_\beta, s_{\alpha}) \cap R(\alpha,\pp)$ incident to $v_{i}$. 
		If $v_{i+1} \in J(s_\beta, \cdot)\cap J(s_\beta, s_{\alpha})$, 
		then $e_{i+1} \subseteq J(s_\beta,\cdot)$.
		If $v_{i+1}\in \Gamma $, then $e_{i+1} \subseteq \Gamma$. 
		(In Fig.~\ref{fig:updatediagram},  
		see $e_{i} = e', v_{i}=z,v_{i+1} = z'$.)		
		\item
		If $ e_{i} \subseteq \Gamma$, let $g$ be the
                $\Gamma$-arc in $\pp$ incident to $v_i$.
		Let $e_{i+1} \subseteq J(s_\beta,s_\gamma)$, where $R(\gamma,\pp)$ is the first 
		region, incident to $g$ clockwise from $v_{i}$ 
		such that $J(s_\beta, s_\gamma)$ intersects $g \cap \overline{R(\gamma,\pp)}$; 
		let $v_{i+1}$  be this intersection point.
		(In Fig.~\ref{fig:updatediagram}, see $v_i = v$ and $v_{i+1} = w$.)
     \end{enumerate}
\end{definition}

The following theorem shows that $J(\beta)$ forms an 
$s_{\beta}$-monotone path joining the endpoints
 of $\beta$.
 We defer its proof to the end of this section.

\begin{figure}
	\centering
	\begin{minipage}{0.48\textwidth}
		\centering
		\includegraphics{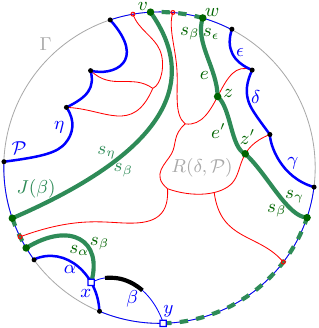}
		\caption{The merge curve $J(\beta)$ (thick, green) on $\vld(\pp)$ (thin, red).}
		\label{fig:updatediagram}
	\end{minipage}
	\hfill
	\begin{minipage}{0.48\textwidth}
		\begin{minipage}{0.99\textwidth}
			\centering
			\includegraphics{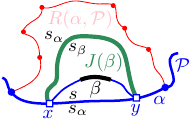}
			\caption{If $\beta$ splits $\alpha$, 
				$J(\beta) \subset R(\alpha,\pp)$ would yield 
				a forbidden  $s_\alpha$-inverse cycle.	}
			\label{fig:connected}
		\end{minipage}
		\begin{minipage}{0.99\textwidth}
			\vspace{19pt}
			\centering
			\includegraphics[page=1]{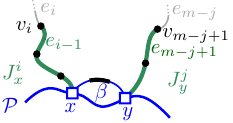}
			\caption{$J_x^i$ and $J_y^j$ in Section~\ref{sec:mergecorrect}.}
			\label{fig:partsofmergecurve}
		\end{minipage} 
	\end{minipage}
\end{figure}

\begin{theorem} \label{thm:mergecorrect}
	The merge curve $J(\beta)$ is a unique  $s_{\beta}$-monotone path in the
	arrangement of $s_\beta$-related bisectors $\arr(\jj{s_\beta}{S_\pp} \cup \Gamma)$ 
	connecting the endpoints of~$\beta$.
        If arc $\beta$ splits a single arc $\alpha\in \pp$
        (case~(\ref{item:splitregion}), Observation~\ref{obs:insert-beta})
        then $J(\beta)$ must intersect $R(\alpha,\pp)$ in two different components,
        $e_1,e_{m-1} \subseteq J(s_\alpha,s_\beta)$.
        $J(\beta)$ can intersect any other region in $\vld(\pp)$
        at most once.
        $J(\beta)$ cannot intersect the region of any
        arc in $\pp\setminus \pp_\beta$, which gets deleted by the insertion of $\beta$,
        nor can it intersect arc $\beta$  in
        its interior.
\end{theorem}

Let  $T(\beta)$ denote the portion of  $\vld(\pp)$ enclosed by
$J(\beta)$.
 $\ins$ is obtained from $\vld(\pp)$ by
deleting $T(\beta)$ and substituting it by $J(\beta)$, i.e.,
$\ins=\big(\vld(\pp) \setminus T(\beta)\big)
\cup J(\beta)$.

\begin{theorem}\label{thm:insertion}
	$\ins$ is the  Voronoi-like diagram 
	 $\vld(\pp_\beta)$.
\end{theorem}

\begin{proof}
  By construction, $\ins$ induces a subdivision of the domain
  $D_{\pp_\beta}$. Let $R(\alpha)$ denote the face of $\ins$  incident
  to a boundary  arc
  $\alpha\in \pp_\beta$.
  By Theorem~\ref{thm:mergecorrect}, $J(\beta)$,
   and thus,  $\partial R(\beta)\setminus\beta$,
  is an
  $s_\beta$-monotone path connecting the endpoints of $\beta$.
        For any arc $\alpha\in \pp$ such that $J(\beta)$ passes through
        $R(\alpha,\pp)$, the boundary of the updated  face 
        in $\ins$ remains an $s_\alpha$-monotone path,
        by the definition of $J(\beta)$.
        Thus, $\partial R(\alpha)\setminus \alpha$ 
        is an
         $s_\alpha$-monotone path for any region
        $R(\alpha)$ in $\ins$, satisfying the first requirement of
        Definition~\ref{def:vld}.
        
        Since $J(\beta)$  can enter any region in $\vld(\pp)$ at
        most once (except case~(c), Observation~\ref{obs:insert-beta}))
        it cannot \emph{cut out} 
        a face that may remain in the interior of $D_\pp$.
        In addition, $J(\beta)$  cannot pass through 
       any  region of an arc  in $\pp\setminus\pp_\beta$,
       thus, such a region must be  enclosed by $J(\beta)$ and will be  deleted.
        Hence, any face of $\ins$ must be incident to a boundary arc
        of $\pp_\beta$, satisfying also the second requirement of
        Definition~\ref{def:vld}.
        Since,  by
        Theorem~\ref{thm:vldunique}, the Voronoi-like diagram  of a boundary
        curve is unique, it follows that  $\ins=\vld(\pp_\beta)$.     
 \end{proof}

The tracing of the merge curve $J(\beta)$ within $\vld(\pp)$,  given the endpoints
of~$\beta$, can be done in linear time similarly to tracing such a curve in any
ordinary Voronoi diagram, see, e.g.,~\cite[Ch.~7.5.3]{Aurenbook}.
This is correct as a result of the cut
property of Lemma~\ref{lem:propertyR}.
When $J(\beta)$ enters a region $R(\gamma,\pp)$ at a
point $v_i$,
we can determine $v_{i+1}$ by scanning $\partial R(\gamma,\pp)$
counterclockwise sequentially until we encounter the first
intersection with $J(s_\beta, s_\gamma)$. 
Lemma~\ref{lem:propertyR} assures that
no  intersection of $J(s_\beta, s_\gamma)$ with $\partial
R(\gamma,\pp)$ 
between $v_i$ and  $v_{i+1}$ is possible,  such as the one shown in Fig.~\ref{fig:tracing}.
Thus, we can state the following fact.

\begin{lemma}\label{lem:scan}
  Let $e_i=(v_i,v_{i+1})$ be an edge of  $J(\beta)$ in $R(\gamma,\pp)$.
  Given $v_i$, we can determine $v_{i+1}$ by sequentially scanning $\partial
  R(\gamma,\pp)$  counterclockwise from $v_i$ (i.e., away from $\gamma$)
  until the first intersection of $J(s_\beta, s_\gamma)$ with $\partial
  R(\gamma,\pp)$ which determines $v_{i+1}$.
 \end{lemma} 

\begin{figure}
	\centering
	\includegraphics{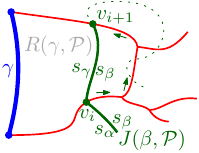}
	\caption{Impossible configuration of $J(s_\beta, s_\gamma)$. Scanning $\partial R(\gamma,\pp)$ from $v_i$ counterclockwise, 
		Lemma~\ref{lem:propertyR} assures that $v_{i+1}$ is the first encountered intersection 
		of $J(s_\beta, s_\gamma)$ with $\partial R(\gamma,\pp)$. }
	\label{fig:tracing}
\end{figure}

Special care is required in Observation~\ref{obs:insert-beta}, cases~(\ref{item:splitregion}),  (\ref{item:splitgap}),
(\ref{item:deletegaps}),
to  identify the first edge of $J(\beta)$ because
$\beta$ does not overlap any  feature of $\vld(\pp)$.
To handle them we need some parameters as defined below.

Let $\tilde\pp$ denote the finer version of $\pp$ derived by
intersecting its $\Gamma$-arcs with  $\vld(\pp)$,
i.e., partitioning  the  $\Gamma$-arcs of $\pp$
into finer pieces 
by the incident faces of $\vld(\pp)$.
Since the complexity of $\vld(\pp)$ is $O(|\pp|)$, it follows that 
$|\tilde\pp|$ is also $O(|\pp|)$.
        
\begin{definition}
  \label{def:parameters}
		Let  $\alpha$ and $\gamma$ denote the original arcs preceding and
                following $\beta$ on $\pp_\beta$.
                We assume a counterclockwise traversal of  $\pp$ and $\pp_\beta$.
                \begin{enumerate}\itemsep=-2pt\topsep=0pt
                  
        	\item Let $d_1(\beta,\pp_\beta)$ denote the number of auxiliary arcs 
        	 that appear on  $\pp_\beta$ from $\alpha$ to
                 $\beta$ (or equivalently from $\beta$ to $\gamma$). 
        	        
        	\item Let $d_2(\beta,\pp_\beta)$ denote the number of
                  auxiliary arcs that appear on  $\pp $ between the
                  endpoints of $\beta$, which get
        	deleted by the insertion of $\beta$.

      \item
        In case (\ref{item:splitregion}) of
        Observation~\ref{obs:insert-beta}, where $\beta$ splits an arc $\omega$ in two arcs $(\omega_1,\omega_2)$,
        let $r(\beta,\pp_\beta)= 
        \min\{|\partial R(\omega_1, \pp_\beta)|, |\partial R(\omega_2,
        \pp_\beta)|\}$; otherwise, let $r(\beta,\pp_\beta)=0$.

       	\item In case~(\ref{item:splitgap}) of Observation~\ref{obs:insert-beta}, 
where $\beta$ splits a $\Gamma$-arc,  let
$\tilde d(\beta,\pp_\beta)$ denote the number of
fine $\Gamma$-arcs  on  $\tilde\pp_\beta$ from $\alpha$ to $\beta$ (i.e., the number of
regions in $\vld(\pp_\beta)$ incident to $\Gamma$ from $\alpha$ to
$\beta$); in all other cases, $\tilde d(\beta,\pp_\beta) = 0$.
    \end{enumerate}
\end{definition}

     	\begin{figure}
     	\centering
     	\begin{minipage}{0.48\textwidth}
     		\centering
     		\includegraphics[page=1]{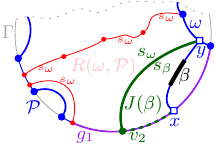}
     		\caption{
     			Case (\ref{item:deletegaps}) of Observation~\ref{obs:insert-beta}, 
     			where  $T(\beta)$ has no leaf on $\pp$.
     			Endpoint $x$ lies on a fine $\Gamma$-arc
     			$g_1$ bounding $R(\omega,\pp)$, 
     			and $y \in \omega$. 
     		}
     		\label{fig:insertinshrunkgap}
     	\end{minipage}
     	\hfill
     	\begin{minipage}{0.48\textwidth}
     		\centering
     		\includegraphics[page=1]{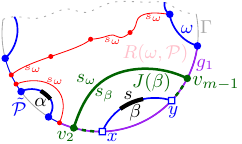}
     		\caption{
     			Case (\ref{item:splitgap}) of Observation~\ref{obs:insert-beta}, 
     			where  $T(\beta)$ has no leaf on $\pp$.
     			Both $x,y$ lie on a fine $\Gamma$-arc 
     			$g_1$ bounding $R(\omega,\pp)$. 
     		}
     		\label{fig:insertingap}
     	\end{minipage}
     \end{figure}

\begin{lemma}\label{lem:insertion-time}
	Given $\alpha$, $\gamma$, and $\vld(\pp)$,
	the merge curve $J(\beta)$ can be computed	in time 
        $O( |J(\beta)| + 
        d_1(\beta,\pp_\beta) + 
        d_2(\beta, \pp_\beta) +
        r(\beta, \pp_\beta)+ 
        \tilde d(\beta,\pp_\beta)  ) $. 
\end{lemma}

\begin{proof}
	  We assume a ccw ordering of $\pp$.
	  We first determine the endpoints of~$\beta$ 
	  in time $O(d_1(\beta,\pp_\beta)+ d_2(\beta,\pp_\beta))$
	  by scanning sequentially  the arcs in $\pp$ starting at $\alpha$ and
	  moving ccw (towards $\gamma$)   
	  until the endpoints of $\beta$ are determined.
	  Note that $\beta$ contains $\beta^*$ therefore we can easily identify the
	  correct component of  $J(s,s_\beta)\cap D_\pp$ during the scan, even
	  if $J(s,s_\beta)$ intersects $\pp$ multiple times.
      This scan further determines which case of Observation~\ref{obs:insert-beta}  
	  is relevant. 

		Let $T(\beta)$ denote the portion of $\vld(\pp)$ that is
		enclosed by  $J(\beta)$ and $\pp\setminus \pp_\beta$;
        this is deleted by the insertion of $\beta$. 
        $T(\beta)$ is a plane forest, which by 
        Theorem~\ref{thm:mergecorrect} is incident to the 
        following faces of $\vld(\pp)$: one face for each bisector edge of
        $J(\beta)$, and one face for each auxiliary arc $\alpha'\in
        \pp\setminus\pp_\beta$. 
        The latter number is counted in $ d_2(\beta, \pp_\beta)$.
        We infer that $T(\beta)$ has complexity 
        $O(|J(\beta)| + d_2(\beta, \pp_\beta))$.
    
To compute $J(\beta)$ we  trace $T(\beta)$ in
time  $O(|T(\beta)|)$, as for any ordinary Voronoi diagram,
and this  is possible due to Theorem~\ref{thm:mergecorrect}  and
Lemma~\ref{lem:scan}.
However, we first need to identify one  leaf of $T(\beta)$.

 Suppose first that $T(\beta)$ has a leaf on $\pp$. 
 Then, in all cases of Observation~\ref{obs:insert-beta}, except cases~(\ref{item:splitgap}) 
 and~(\ref{item:deletegaps}), a leaf of $T(\beta)$ is identified  
 by the initial scan. 
In case~(\ref{item:deletegaps}), $\beta$ has at least one endpoint
on a boundary arc $\rho$ of $\pp$, see Fig.~\ref{fig:updatediagram}.
We identify a leaf by scanning $\tilde\pp$ starting at $\rho$
and moving towards the other endpoint of $\beta$.
The scan takes only one step as the leaf will be incident to the first $\Gamma$-arc 
neighboring $\rho$ on $\tilde\pp$. 
In case~(\ref{item:splitgap}) both endpoints of $\beta$ are on
 $\Gamma$.
We scan 
$\tilde \pp$ from $\alpha$ to $\beta$
until we locate the first endpoint of $\beta$, $x$.
A leaf of $T(\beta)$ must be incident to  the fine $\Gamma$-arc  that
contains $x$.
Since the  encountered $\Gamma$-arcs 
remain in $\tilde \pp_\beta$, the term $O(\tilde d(\beta,
\pp_\beta))$ is  added to the overall time complexity.

Suppose now that  $T(\beta)$ has no leaf on $\pp$. Then  $\beta$ is
enclosed within a single Voronoi-like region $R(\omega,\pp)$.
There are three cases Observation~\ref{obs:insert-beta}~(\ref{item:splitregion}),
(\ref{item:splitgap}), and (\ref{item:deletegaps}) to consider.

        In  case Observation~\ref{obs:insert-beta}(\ref{item:splitregion}), the insertion of
        $\beta$ splits arc $\omega$ in two parts, $\omega_1$ and
        $\omega_2$. 
        To identify a leaf of $T(\beta)$
        we scan $\partial R(\omega,\pp)$ sequentially until
        an intersection with  $J(s_{\omega},s_{\beta})$ is found.
        We start scanning from both endpoints of
        $\omega$, tracing the shorter among 
        $\partial R(\omega_1,\pp_\beta)$ and  
        $\partial R(\omega_2, \pp_\beta)$. This adds the term 
		$r(\beta,\pp_\beta)$ to the overall time complexity.

In cases Observation~\ref{obs:insert-beta}(\ref{item:splitgap}),(\ref{item:deletegaps}),
$J(\beta) \subseteq R(\omega,\pp) \cup \Gamma$,
since otherwise $J(\beta)$ would intersect the 
region $R(\omega,\pp)$ twice,   
contradicting Theorem~\ref{thm:mergecorrect}.
Thus, $J(\beta)$ consists of a single bisector  $J(s_\omega,s_\beta)$
and one ($m=3$) or two ($m=4$) $\Gamma$-arcs
see Figs.~\ref{fig:insertinshrunkgap} and~\ref{fig:insertingap}. 
Thus, it is enough  to identify $\omega$. 
In case~(\ref{item:deletegaps}), $\omega$ is identified during the
initial scan. 
In case (\ref{item:splitgap}),
$\beta$ has both endpoints on~$\Gamma$.
We scan $\tilde \pp$ from $\alpha$ to $\beta$
until we locate the first endpoint of $\beta$, $x$.
Then the $\Gamma$-arc that contains $x$ in $\tilde\pp$
bounds the region  $R(\omega,\pp)$.
This scan adds the term $O(\tilde d(\beta, \pp_\beta))$ to the time complexity.
 \end{proof}

\subsection{Proving Theorem~\ref{thm:mergecorrect}}
\label{sec:mergecorrect}

We first establish the following lemma.

\begin{lemma}
$J(\beta)$ cannot intersect arc $\beta$, other
than its endpoints. 
\end{lemma}

\begin{proof}
Suppose that an edge $e_i$ of $J(\beta)$, such that $e_i\subseteq
J(s_\alpha,s_\beta) $ and $e_i \in  R(\alpha,\pp)$,
intersects arc $\beta$. Then $J(s,s_\alpha)$ must also pass through the same
intersection point within $R(\alpha,\pp)$. But an 
 $s$-bisector $J(s,s_\alpha)$ can never intersect $R(\alpha,\pp)$, by
Lemma~\ref{lem:regionclosertos}. 
\end{proof}

We use the following observation throughout the proofs in this section.
      
\begin{lemma}\label{lem:domain}
  For any $p\in S$, 
	$D(s,p) \cap D_\pp $ is connected.
	Thus, any components of the same $s$-bisector 
	$J(s,\cdot) \cap  D_\pp$ must appear sequentially along $\pp$.
\end{lemma}
\begin{proof}
	If we assume the contrary we obtain a forbidden $s$-inverse 
	cycle defined by $J(s,\cdot)$ and $\pp$.
 \end{proof}

We now establish that $J(\beta)$ cannot pass through
any region  of  auxiliary arcs in $\pp\setminus \pp_\beta$ that get deleted by the insertion of
$\beta$.

\begin{lemma} \label{lem:propertyRs}
	Let $\alpha\in \pp$ but $ \alpha\not\in \pp_\beta$. 
	Then $R(\alpha,\pp) \subset D(s_\beta,s_{\alpha})$. 
\end{lemma}
\begin{proof}
	\begin{figure}
		\centering
		\includegraphics[page=1]{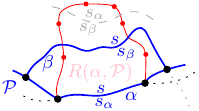}
		\caption{Illustrations for Lemma~\ref{lem:propertyRs}.
		}
		\label{fig:beta_deletes_alpha}
	\end{figure}
	By Lemma~\ref{lem:regionclosertos}, it holds that $R(\alpha,\pp) \subseteq D(s,s_\alpha)$.
	Let $R_s=R(\alpha,\pp) \cap D(s,s_\beta)$ and $R_\beta = R(\alpha,\pp) \cap D(s_\beta,s)$.
	By transitivity of dominance regions we have  $R_\beta \subseteq D(s_\beta,s_\alpha)$. 
	By Lemma~\ref{lem:domain}, $R_s$ is not incident to $\alpha$. 
	Thus, if  $J(s_\beta,s_\alpha)$ intersected $R_s$  then
	it would create a forbidden $s_\alpha$-cycle 
	contradicting Lemma~\ref{lem:nocycleindomain},
	see the dashed gray line in Fig.~\ref{fig:beta_deletes_alpha}.
	This implies that also $R_s \subseteq D(s_\beta,s_\alpha)$. 
	Thus,  $R(\alpha,\pp) = R_s \cup R_\beta \subseteq D(s_\beta,s_\alpha)$.
 \end{proof}

In the following we prove that $J(\beta)$ is an $s_\beta$-monotone path connecting the endpoints 
of $\beta$.
To this aim we perform a bi-directional induction on the vertices of
$J(\beta)$.

Let  $J_x^i = (v_1, v_2, \ldots, v_i), 1\leq i < m$, be the subpath of $J(\beta)$
starting at $v_1=x$ up to vertex $v_i$, including a small neighborhood 
of $e_{i}$ incident to $v_i$, see Fig.~\ref{fig:partsofmergecurve}. 
Note that vertex $v_i$ uniquely determines $e_{i}$, however,
its other endpoint is not yet specified.
Similarly, let $J_y^j = (v_m, v_{m-1}, \ldots, v_{m-j+1}),1\leq j < m$, 
denote  the subpath of $J(\beta)$, starting at $v_m$ up
to vertex $v_{m-j+1}$, including a small neighborhood of edge
$e_{m-j}$. 
For any bisector edge $e_\ell\in J(\beta)$,
let $\alpha_\ell$
denote the boundary arc that induces $e_\ell$, i.e.,
$e_\ell \subseteq J(s_{\alpha_\ell},s_\beta) \cap R(\alpha_\ell,\pp)$.

\emph{Induction hypothesis}: Suppose $J_x^i$ and  $J_y^{j}$, $i,j\geq 1$, are
disjoint  $s_\beta$-monotone paths.
Suppose further that each bisector edge of $J_x^i$ and of $J_y^j$ 
passes through a distinct
region of $\vld(\pp)$: $\alpha_\ell$ is  distinct for $\ell$,  
$1\leq\ell\leq i$ and  $m-j\leq \ell<m$, except possibly
$\alpha_i=\alpha_{m-j}$ and $\alpha_1=\alpha_{m-1}$.

\emph{Induction step}:
Assuming that $i+j<m$, we prove that at least one of $J_x^i$ or $J_y^j$
can respectively grow to $J_x^{i+1}$ or $J_y^{j+1}$ at a valid vertex
(Lemmas~\ref{lem:vertexonarc}, \ref{lem:gammabad}),
and it enters a new region of $\vld(\pp)$ that has not been visited so
far (Lemma~\ref{lem:distinctregions}).
A valid vertex belongs in  $\arr(\jj{s_\beta}{S_\pp} \cup \Gamma)$
or is an endpoint of $\beta$.
A finish condition when $i+j=m$ is given in Lemma~\ref{lem:finish}.
The base case for $i=j=1$ is trivially true.

Suppose that $e_i\subseteq  J(s_{\alpha_i},s_\beta)$ and
$v_i\in \partial R(\alpha_i,\pp)$.
To show that $v_{i+1}$ is a valid vertex it  is enough to show that
(1)  $v_{i+1}$ can not be on $\alpha_i$, and (2)
if $v_i$ is on a $\Gamma$-arc 
then $v_{i+1}$ can be determined on the same $\Gamma$-arc.
However, we cannot easily derive these conclusions directly. Instead we
show that if  $v_{i+1}$ is not valid then $v_{m-j}$ will have to be
valid.

In the following lemmas we assume that the  induction hypothesis holds.

\begin{lemma}\label{lem:vertexonarc}
	Suppose $e_i\subseteq  J(s_{\alpha_i},s_\beta)$ but
        $v_{i+1}\in \alpha_i$, that is, $e_i$ hits
        arc  $\alpha_i\in \pp$, and
        thus, $v_{i+1}$ is not a valid vertex.
	Then vertex $v_{m-j}$ must be a valid vertex in $\arr(\jj{s_\beta}{S_\pp})$, 
	and $v_{m-j}$ can not be on $\pp$.
\end{lemma}

\begin{proof}
	\begin{figure}
		\centering
		\includegraphics{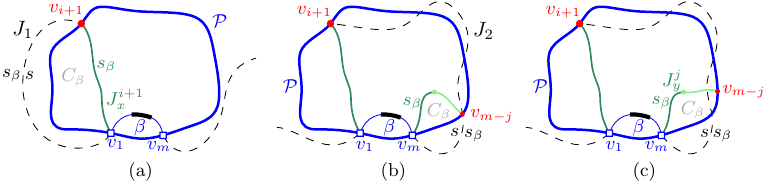}
		\caption{The assumption that edge $e_i = (v_i,v_{i+1})$ of 
			the merge curve $J_x^i$ hits a boundary arc of $\pp$ 
			as in Lemma~\ref{lem:vertexonarc}. }
		\label{hitboundary}
	\end{figure}
	Suppose vertex $v_{i+1}$ of $e_{i}$ lies on arc $\alpha_i$ 
	as shown in Fig.~\ref{hitboundary}(a).
	Vertex $v_{i+1}$ is the intersection point of  related bisectors $J(s, s_{\alpha_i})$, 
	$J(s_\beta, s_{\alpha_i})$ 
	and thus also of $J(s,s_\beta)$.
	Thus, $v_1,v_m,v_{i+1} \in J(s,s_\beta)$.
        By the induction hypothesis, no other vertex of 
        $J_x^i$ nor  $J_y^j$ can be on $J(s,s_\beta)$.
        Vertices $v_1,v_{i+1},v_m$ appear on $\pp$ in clockwise order, 
	because $J_x^{i+1}$ cannot intersect $\beta$. 
        Arc $\beta$ partitions $J(s,s_\beta)$ in two parts: 
	$J_1$ incident to $v_1$ and $J_2$ incident to $v_m$.
     We claim that $v_{i+1}$ must lie on $J_2$, as  otherwise, 
	$J_x^{i+1}$ and $J_1$ would form a forbidden $s_\beta$-inverse cycle, 
	see the dashed black and the green solid curve in Fig.~\ref{hitboundary}(a), 
	contradicting Lemma~\ref{lem:nocycle}.
        This cycle must be $s_\beta$-inverse  because $J_x^{i+1} \subseteq \overline{D_\pp}$,  and all components of 
		$J(s,\cdot) \cap  D_\pp$ must appear sequentially
        along $\pp$ by Lemma~\ref{lem:domain}.
                
	Thus,  $v_{i+1}$ lies on $J_2$. 
	Further, by  Lemma~\ref{lem:domain}, the components of $J_2 \cap D_\pp$ appear 
	on $\pp$ clockwise after 
	$v_{i+1}$ and before $v_m$, 
	as shown in Fig.~\ref{hitboundary}(b), which  illustrates  $J(s,s_\beta)$ 
	as a black dashed  curve.

	Now consider $J_y^{j}$.  We show that $v_{m-j}$ cannot be on $\pp$.
	First observe that $v_{m-j}$ can not lie on $\pp$,
	clockwise after $v_m$ and before $v_1$, 
	since $J_y^{j+1}$ cannot cross $\beta$.
	Now we prove that  $v_{m-j}$ cannot lie on $\pp$ clockwise 
	after $v_1$ and before  $v_{i+1}$.
	To see that, note that edge $e_{m-j}$ cannot cross any 
	non-$\Gamma$ edge of $J_x^{i+1}$, because by the induction hypothesis, 
	$\alpha_{m-j}$ is distinct from all $\alpha_\ell, \ell \leq i $.
	In addition, by the definition of a $\Gamma$-arc, 
	$v_{m-j}$ cannot lie on any $\Gamma$-arc  of $J_x^{i}$. 
	Finally, we show that $v_{m-j}$ cannot lie on $\pp$  
	clockwise after $v_{i+1}$ and before $v_m$.
	If $v_{m-j}$ lay on the boundary arc $\alpha_{m-j}$ 
	then we would have $v_{m-j} \in J(s,s_\beta)$.
	This would define an $s_\beta$-inverse cycle $C_\beta$, formed by 
	$J_y^{j+1}$ and $J(s_\beta, s)$, see
        Fig.~\ref{hitboundary}(b), similarly to 
        the first paragraph of this proof.
	If $v_{m-j}$ lay on a $\Gamma$-arc 
	then there would also be a forbidden $s_\beta$-inverse cycle formed by 
	$J_y^{j+1}$ and $J(s,s_\beta)$ because in order to reach $\Gamma$, 
	edge $e_i$ must cross $J(s,s_\beta)$.
	See the dashed black and the 
	green curve in Fig.~\ref{hitboundary}(c). 
	Thus $v_{m-j} \not \in \pp$.
	
	Since $v_{m-j}\in \partial R(\alpha_{i+1})$ but $v_{m-j}\not\in\pp$, it must be   
	a vertex of $\arr(\jj{s_\beta}{S_\pp})$. 
 \end{proof}

The proof for the following lemma is similar.

\begin{lemma}\label{lem:gammabad}
	Suppose vertex $v_i$ is on a $\Gamma$-arc $g\in \pp$ but $v_{i+1}$ cannot
	be determined  because no bisector $J(s_\beta, s_\gamma)$ intersects 
	$\overline{R(\gamma,\pp)}\cap g$, clockwise from $v_i$. 
	Then vertex $v_{m-j}$ must be a valid vertex in
	$\arr(\jj{s_\beta}{S_\pp})$ and $v_{m-j}$ can not be on $\pp$.
\end{lemma}
\begin{proof}
	\begin{figure}
		\centering
		\includegraphics{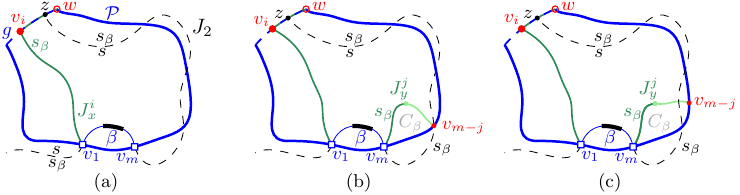}
		\caption{The assumption that $v_i \in \Gamma$ and $v_{i+1}$ of 
			the merge curve $J_x^i$ cannot be determined 
			as in Lemma~\ref{lem:gammabad}. }
		\label{fig:gammabad}
	\end{figure}
	We truncate the $\Gamma$-arc $g$ to its portion clockwise from
	$v_i$; let  $w$ be the endpoint of $g$ clockwise from $v_i$, see
	Fig.~\ref{fig:gammabad}(a).
	If no $J(s_\beta, s_\gamma)\cap R(\gamma,\pp)$  intersects $g$, as
	we assume in this lemma, then $R(\gamma,\pp)\cap g\subseteq
	D(s_\beta, s_\gamma)$, for any  region $R(\gamma,\pp)$ incident
	to $g$. Thus, $w \in D(s_\beta, s)$.
	However, $v_i \in D(s,s_\beta)$, since, by
	Lemma~\ref{lem:regionclosertos}, $R(\alpha_{i-1})\subseteq
	D(s,s_{\alpha_{i-1}})$ and $v_i$ is incident to $J(s_\beta,
	s_{\alpha_{i-1}})\cap R(\alpha_{i-1})$. 
	Thus, $J(s,s_\beta)$ must intersect $g$ at some point $z$ 
	clockwise from $v_i$.
	Arc $\beta$ partitions  $J(s,s_\beta)$ in two parts: $J_1$  incident to $v_1$ and  $J_2$ incident
	to $v_m$. 
	Lemma~\ref{lem:domain} implies that 
	all components of $J_2 \cap D_\pp$ appear on $\pp$ clockwise after 
	$v_{i}$ and before $v_m$,
	as shown by the black dashed  curve in Fig.~\ref{fig:gammabad}(a);
	also $z$ lies on $J_2$. 
	
	Now we can show that vertex $v_{m-j}$ of $J_y^{j}$ cannot be on $\pp$
	analogously to the proof of Lemma~\ref{lem:vertexonarc}. 
	The only difference is that we must additionally show that $v_{m-j}$ cannot lie on $\pp$ clockwise after 
	$v_i$ and before $w$. But this holds already by the assumption
        in the lemma statement.
	Refer to Figures~\ref{fig:gammabad}(b) and (c).
	
	We conclude that $v_{m-j}$ cannot lie on $\pp$ and it is  
	a valid vertex of $\arr(\jj{s_\beta}{S_\pp})$.
 \end{proof}

Lemma~\ref{lem:finish} in the sequel provides a finish condition for
the induction, when $J_x^i$ and $J_y^j$ are 
incident to a common region or to a common 
$\Gamma$-arc. 
When it is met,
the merge curve $J(\beta)$ 
is a concatenation of $J_x^i$ and $J_y^j$.

\begin{lemma}\label{lem:finish}
  Suppose $i+j>2$ and  either (1) or (2) holds:
  (1) $v_i$ and $v_{m-j+1}$ are incident to a common region
  $R(\alpha_i,\pp)$ and $e_i,e_{m-j}\subseteq J(s_\beta, s_{\alpha_i})$,
  i.e., $\alpha_i=\alpha_{m-j}$;
	or (2) $v_i$ and $v_{m-j+1}$ are on 
	a common $\Gamma$-arc $g$ of $\pp$ and 
	$e_i, e_{m-j} \subseteq \Gamma$.
	Then $v_{i+1}=v_{m-j+1}$, $v_{m-j}=v_i$, and $m=i+j$.
\end{lemma}

	\begin{figure}
	\centering
	\includegraphics{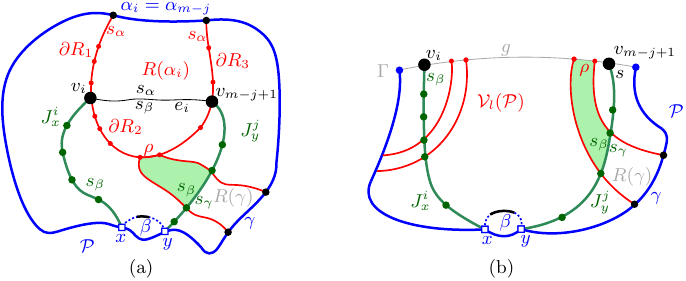}
	\caption{Illustrations for Lemma~\ref{lem:finish}. 
		(a) corresponds to condition (1)  and (b) to condition~(2). 
	Here, the label $R(\gamma)$ abbreviates $R(\gamma, \pp)$.}
	\label{fig:incidentsameregion}
\end{figure}

\begin{proof}	
	Let $\alpha = \alpha_i$. 
	Suppose (1) holds, then $e_i,e_{m-j}\subseteq J(s_\beta, s_{\alpha})$,
	see Fig.~\ref{fig:incidentsameregion}(a). 
	The boundary $\partial R(\alpha_i,\pp)$ is partitioned in four parts, using a
	counterclockwise traversal starting at $\alpha_{i}$:
	$\partial R_1$, from the endpoint of arc $\alpha_i$ to $v_i$;
	$\partial R_2$, from  $v_i$ to $v_{m-j+1}$;
	$\partial R_3$, from  $v_{m-j+1}$ to the next endpoint of
	$\alpha_i$;	and arc $\alpha_i$.
	We show that $e_i$ and $e_{m-j}$ cannot hit any of these
	parts; thus, $e_i=e_{m-j}$.
	\begin{enumerate}
          \itemsep=0pt\topsep=0pt
		\item 
		Edge $e_i$ cannot hit $\partial R_1$ and edge $e_{m-j}$ cannot
		hit $\partial R_3$ by the cut property, Lemma~\ref{lem:propertyR}.
		\item
		We prove that edge $e_i$ cannot hit
		$\partial R_2$. Analogously for edge $e_{m-j}$. 
		Let $\rho$ be any edge on $\partial R_2$. 
		(If $v_i \in \rho$ or $v_{m-j+1} \in \rho$, assume that $\rho$
		is truncated with endpoint $v_i$ or $v_{m-j+1}$ respectively).
		\begin{enumerate}
                  \itemsep=0pt
			\item	Suppose that $\rho$ is a bisector  edge, 
			$\rho \subseteq J(s_\alpha,s_\gamma)$, see Fig.~\ref{fig:incidentsameregion}(a).
			Then at least one of $J_y^j$, $J_x^i$, or $\beta$ must pass through
			$R(\gamma,\pp)$.  
			Suppose that $J_y^j$ does, as shown in
			Fig.~\ref{fig:incidentsameregion}(a).
			Then by the cut property (Lemma~\ref{lem:propertyR})
                      $\rho
			\subseteq D(s_\beta,s_\gamma)$.
			By transitivity (Lemma~\ref{lem:transitivity}) it also holds 
			that $\rho \subseteq D(s_\beta,s_\alpha)$.
			Thus,  $e_i$ cannot hit $\rho$.
			Symmetrically for $J_x^i$.
			If only $\beta$ passes through $R(\gamma,\pp)$, then we can use
			Lemma~\ref{lem:propertyRs} to derive that $\rho
			\subseteq D(s_\beta,s_\gamma)$; the rest follows. 
			
			\item 	Suppose that $\rho \subseteq \Gamma$.
			Then either $\rho$ itself is part of an edge of $J_y^j$ or of  $J_x^i$, or
			$\beta$ passes through $R(\alpha,\pp)$ and $\rho$  is at opposite side
			of it than $\alpha$.
			In the former case, 	$\rho \subseteq D(s_\beta,s_\alpha)$ by
			the definition of a $\Gamma$-edge in the merge curve. 
			In the latter case, the same is derived by
			Lemma~\ref{lem:regionclosertos}
			and transitivity
			(Lemma~\ref{lem:transitivity}).
			Thus,  $e_i$ cannot hit $\rho$.
		\end{enumerate}

		\item	
		Edge  $e_i$ (resp. $e_{m-j}$) cannot hit
		$\partial R_3$ because if it did,
		$e_i$ and $e_{m-j}$ would not appear sequentially on 
		$R(\alpha_i,\pp)$
		contradicting Lemma~\ref{lem:propertyR}.
		
		\item 
		It remains to show that  $e_i$ and  $e_{m-j}$ cannot both hit $\alpha_i$.
		But this is already shown in  Lemma~\ref{lem:vertexonarc}.
	\end{enumerate}
	
	Now suppose (2) holds, 
	see Fig.~\ref{fig:incidentsameregion}(b). 
	Let $R(\gamma,\pp)$ be a region in $\vld(\pp)$ incident to the
        $\Gamma$-arc $g$  
and let $\rho= R(\gamma,\pp)\cap g$ be the $\Gamma$-arc bounding $R(\gamma,\pp)$, 
which lies between $v_i$ and $v_{m-j+1}$.
	At least one of $J_y^j$ or $J_x^i$ or $\beta$ must pass 
	through $R(\gamma,\pp)$.
	By the exact same arguments as before, 
	$\rho \subseteq D(s_\beta,s_\gamma)$.
	We infer that there is no bisector  
	$J(s_\beta,s_\gamma)$ in $R(\gamma,\pp)$, for any region
	$R(\gamma,\pp)$ incident to $g$ between $v_i$ and $v_{m-j+1}$.
	Thus, $e_{i+1}=e_{m-j+1}\subseteq g$.
	
	Thus, in both (1) and (2) $v_{i+1}=v_{m-j+1}$, $v_{m-j}=v_i$, and $m=i+j$.
	$J(\beta)$ is the concatenation of $J_x^i$ and $J_y^j$ with  $e_{i+1}=e_{m-j+1}$.
 \end{proof}

\begin{lemma}\label{lem:distinctregions}
	Suppose vertex $v_{i+1}$ is valid and $e_{i+1}\subseteq
	J(s_\beta,s_{a_{i+1}})$. Then $R(\alpha_{i+1})$ has not been visited
	by $J_x^i$ nor $J_y^j$, i.e., $\alpha_{i+1}\neq \alpha_\ell$ for $\ell\leq i$ and for 
	$ m-j<\ell$.
\end{lemma}

	\begin{figure}
	\centering
	\includegraphics{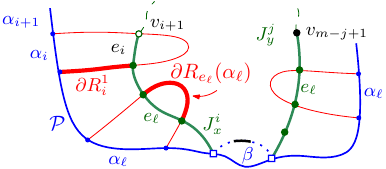}
	\caption{Illustration for Lemma~\ref{lem:distinctregions}.
	}
	\label{fig:distinct_regions}
\end{figure}

\begin{proof}
	Let $e_k, k \leq i$, be a bisector edge of $J_x^i$.  
	Denote by $\partial R_k^1$ the portion of $\partial R(\alpha_k)$ from
	$\alpha_k$ to $v_k$ in a counterclockwise traversal, 
	see the bold red part $\partial R_i^1$ in Fig.~\ref{fig:distinct_regions}.
	Analogously for a bisector edge $e_{m-j}$ of $J_y^j$, where $\partial R_{m-j}^1$ is
	defined in a clockwise traversal of $\partial R(\alpha_{m-j})$.
	Recall that $\partial R_{e_k}(\alpha_k)$, 
	denotes the portion of $\partial R(\alpha_k)$ 
	\emph{cut out} by edge $e_k$, at
	opposite side from $\alpha_k$. 
	
	The cut property of Lemma~\ref{lem:propertyR} implies that   
	$v_{i+1}$ cannot be on $\partial R_{e_\ell}(\alpha_\ell)$ 
	for any $\ell$, $\ell< i$ and  $ m-j<\ell$ and that 
	$v_{i+1}$ cannot be on $\partial R_{i}^1$.
	This implies that $v_{i+1}$ cannot be on $\partial R_{\ell}^1$ for any
	$\ell< i$, because we have a plane graph in $D_\pp$  and by its layout 
	$\partial R_\ell^1$ is not reachable from $e_{i}$ without first hitting
	$\partial R_{e_\ell}(\alpha_\ell)$ 
	or $\partial  R_i^1$. 
	See Fig.~\ref{fig:distinct_regions}. 
	Thus, $v_{i+1}$ can not be on $\partial R(\alpha_\ell)$, $\ell <i$.
	By Lemma~\ref{lem:finish} $v_{i+1}$ cannot be on $\partial
	R_{m-j}^1$. 
	This implies, again by the layout, 
	that $v_{i+1}$ cannot be on $\partial R_{\ell}^1$ for all $\ell >m-j$.
	Thus, $v_{i+1}$ can not be on $\partial R(\alpha_\ell)$, for any 
	$\ell >m-j$.
	This implies that $\alpha_{i+1}\neq \alpha_\ell$, for any $\ell$,
	$\ell \leq i$ or $\ell >m-j$.
 \end{proof}

By Lemma~\ref{lem:distinctregions},
$J_x^{i+1}$ and $J_y^{j+1}$ always enter a new region of
$\vld(\pp)$ that has not been visited in any previous step.  Hence, conditions (1) or (2)
of Lemma~\ref{lem:finish} must be fulfilled at some point of the
induction, completing the proof of Theorem~\ref{thm:mergecorrect}.

Completing the bi-directional induction establishes also the
remaining properties for $J(\beta)$.
First, $J(\beta)$ can never enter the same region twice
(by Lemma~\ref{lem:distinctregions}),
except the region of $\alpha_1$, if $\alpha_1=\alpha_m$.
This is Observation~\ref{obs:insert-beta}(\ref{item:splitregion}),
where arc $\beta$ 
splits a single arc $\alpha\in \pp$.
In this case $J(\beta)$  enters  $R(\alpha,\pp)$ exactly twice and
both $e_1,e_{m-1} \subseteq J(s_\alpha,s_\beta)$.
This is because  $J(\beta)$ must intersect $\partial
R(\alpha,\pp)$, i.e., 
$J(\beta) \not\subseteq R(\alpha,\pp)$, as otherwise
$J(\beta)=
J(s_\alpha,s_\beta)$
(see Fig.~\ref{fig:connected})
contradicting the labeling of the cut property
in Lemma~\ref{lem:propertyR}.

Completing the induction for Theorem~\ref{thm:mergecorrect}
establishes also that $J(\beta)$ is unique and that the conditions of  
Lemmas~\ref{lem:vertexonarc} and~\ref{lem:gammabad} can never be
met. Thus, no vertex of $J(\beta)$, except its endpoints, 
can be on a boundary arc of $\pp$.

\section{$\vld(\pp)$ is unique} \label{sec:unique}

In this section we prove Theorem~\ref{thm:vldunique} and establish
that for a boundary curve $\pp$ on
$\arcs'\subseteq \arcs$,  
the Voronoi-like diagram $\vld(\pp)$ is unique.

We first show an  essential property of Voronoi-like
regions, which  completes the \emph{cut property} of
Lemma~\ref{lem:propertyR}.

	\begin{figure}
	\begin{minipage}{0.48\textwidth}
		\centering
		\includegraphics{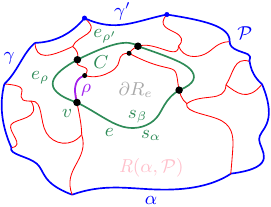}
		\caption{A component $e$ of $J(s_\alpha,\cdot)$ in $R(\alpha,\pp)$ as in Lemma~\ref{lem:missingarc}.}
		\label{fig:cycle}
	\end{minipage}
	\hfill
	\begin{minipage}{0.47\textwidth}
		\centering
		\includegraphics{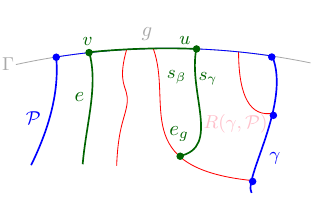}
		\caption{A component $e$ of $J(s_\alpha,\cdot)$ in $R(\alpha,\pp)$ 
			with its endpoint $v$ on a $\Gamma$-arc $g$ 
			as in Lemma~\ref{lem:missingarc}. 
		}
		\label{fig:noentry}
	\end{minipage}
\end{figure}

\begin{lemma}\label{lem:missingarc}
  Let $\beta^*\in \arcs'$ and $\beta^*\subseteq J(s,s_\beta)$.
	Suppose that a component $e$ of
	$J(s_\alpha,s_\beta)$ intersects  $R(\alpha,\pp)$ in $\vld(\pp)$. 
	Then  $J(s, s_\beta)$ must also intersect the domain $D_\pp$. Further,
	there exists a component $\beta$ of $J(s,s_\beta )\cap D_\pp$
	such that the merge curve $J(\beta)$ in $\vld(\pp)$ contains
        $e$ (i.e., $e\subset \partial R(\beta,\pp\oplus\beta)$).
\end{lemma}

\noindent
We say that  arc $\beta$ is \emph{missing from $\pp$}. 

\begin{figure}
	\centering
	\includegraphics{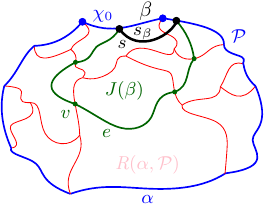}
	\caption{Arc $\beta \subseteq J(s,s_\beta)$ in $D_\pp$;
		The merge curve $J(\beta)$  contains~$e$.}
	\label{fig:missingcurve}
\end{figure}

\begin{proof}
        Suppose there is a non-empty component~$e$ of $J(s_\alpha,s_\beta)
	\cap R(\alpha,\pp)$, however,  $J(s,s_\beta)\cap D_\pp=\emptyset$,
        thus, $D_\pp  \subseteq D(s,s_\beta)$.
        By the transitivity of dominance regions
                        (Lemma~\ref{lem:transitivity}),  it follows
                        that 	for any arc $\chi\in \pp$,
                        $\chi\subseteq D(s_\chi, s_\beta)$.
                   	Let $R_e(\alpha)$ denote the portion of $R(\alpha,\pp)$ cut out by $e$
	(at opposite side from $\alpha$) as defined in
        Lemma~\ref{lem:propertyR}; then 
	 $\partial R_e(\alpha)\subseteq D(s_\beta,s_\alpha)$.

	Consider an endpoint $v$ of $e$. There are two cases: 
	\begin{enumerate}
          \itemsep=0pt\topsep=0pt
		\item 
                If $v$ is on an edge $\rho$ 
                  incident to regions $R(\alpha,\pp)$ and $R(\gamma,\pp)$, then 
		$J(s_\beta,s_{\gamma})$ intersects $R(\gamma,\pp)$ by an edge 
		$e_{\rho}$, incident to $v$,  leaving $\rho$ and $\gamma$ at opposite sides, 
		because $\gamma \subseteq D(s_\gamma, s_\beta)$,
		see Fig.~\ref{fig:cycle}.
		\item 
		If $v$ is on a $\Gamma$-arc $g$, let $R(\gamma,\pp)$ be
		the  first region after $v$ (towards $D(s_\beta,s_\alpha)$)
		with $J(s_\beta,s_{\gamma})$ intersecting $g \cap
		\overline{R(\gamma,\pp)}$ at point $u$, see
		Fig.~\ref{fig:noentry}.  There
		exists such $R(\gamma,\pp)$ because 
		for all boundary arcs $\chi\in \pp$, $\chi\subseteq D(s_\chi,
		s_\beta)$, and this includes the boundary arc that is 
		incident to $g$. 
		Let $e_{g}$ be the component of $J(s_\beta,s_{\gamma})\cap
		R(\gamma,\pp)$ incident to $u$.
	\end{enumerate}
	Thus, given $e$ and $v$, we derive an edge $e'$, either
	$e'=e_\rho$ or $e'=e_{g}$,
	with the same properties as $e$, in a different region of $\vld(\pp)$.
	This process repeats 
	and there is no way to break it because for any arc $\chi\in \pp$,
	$\chi\subseteq D(s_\chi, s_\beta)$.
	Thus, we create  a closed curve on $\vld(\pp)$ consisting of
	consecutive pieces of $J(s_\beta,.)$, possibly interleaved with
	$\Gamma$-arcs, which has the label $s_\beta$ in its interior.
	No two edges of this curve can intersect because otherwise the
	bisector corresponding to such intersecting edges would not be a
	Jordan curve.
        Furthermore,  no
	vertex of this curve can repeat under our general position assumption as no three $s_\beta$-related
        bisectors can intersect at the same point. 
	Thus, the closed curve must be an $s_\beta$-cycle $C$ that is
        contained in $D_\pp$, see Fig.~\ref{fig:cycle}, which
        contradicts Lemma~\ref{lem:nocycleindomain}.
        Thus, our assumption that $J(s,s_\beta)\cap D_\pp=\emptyset$ 
        was false, and thus, $J(s,s_\beta)$ must intersect $\pp$.
        The above process must encounter such an intersection 
        as
        otherwise  the forbidden 
        $s_\beta$-cycle $C$ would exist.
        Let $J_e(\beta)$ denote the sequence of encountered edges $e_\rho$,
	starting with the initial edge $e$ and ending on the  first
        intersection of an arc $\chi_0$ in  $\pp$ with $J(s,s_\beta)$.
	Let  $\beta$ be the component of $J(s,s_\beta)\cap D_\pp$
        incident to  $\chi_0$, 	see Fig.~\ref{fig:missingcurve}.
        
        By its definition, the path $J_e(\beta)$ 
	fulfills the definition of the merge curve $J(\beta)$
	(Definition~\ref{def:mergecurve}).
	Since by Theorem~\ref{thm:mergecorrect} the merge curve
	$J(\beta)$ on $\vld(\pp)$ is unique,
	it follows that $J(\beta)$ 
	contains $J_e(\beta)$, and thus, it also contains edge~$e$.
 \end{proof}

We can now prove  Theorem~\ref{thm:vldunique} from
Section~\ref{sec:problem}.

\begin{apptheorem}{\ref{thm:vldunique}}
	Given a boundary curve $\pp$ for $\arcs'\subseteq\arcs$,
	$\vld(\pp)$ is unique. 
\end{apptheorem}
\begin{proof}
	Let  $\pp$  be a boundary curve for $\arcs'\subseteq\arcs$
	such that  $\pp$  admits a Voronoi-like diagram 
	$\vld(\pp)$.
	Suppose there exist two different Voronoi-like diagrams 
	of $\pp$,  $\vld^{(1)} \neq \vld^{(2)}$. 
	Then there must be an edge $e^{(1)}$ of  $\vld^{(1)}$  
	bounding regions $R^{(1)}(\alpha,\pp)$ and $ R^{(1)}(\beta,\pp)$ of
	$\vld^{(1)}$, where $\alpha,\beta\in \pp$, such that   
	$e^{(1)}$ intersects 
	region $R^{(2)}(\alpha,\pp)$ of $\vld^{(2)}$, since $\alpha$ is
	common to both $R^{(1)}(\alpha,\pp)$ and $R^{(2)}(\alpha,\pp)$. 
	
	Let edge $e \subseteq J(s_\beta, s_\alpha)$ be the component of 
	$R^{(2)}(\alpha,\pp) \cap J(s_\beta, s_\alpha)$ overlapping  with $e^{(1)}$,
	see  Fig.~\ref{fig:unique_vld}.
	From Lemma~\ref{lem:missingarc} it follows that there is a
	non-empty component $\beta_0$ of $J(s,s_\beta)\cap D_\pp$ such
	that $J(\beta_0)$ in $\vld^{(2)}$ contains edge $e$.
	Since $J(\beta_0)$ and $\partial R^{(1)}(\beta,\pp)$ have an overlapping portion 
	$\left(e \cap e^{(1)}\right)$ and they bound the regions of two different arcs 
	$\beta_0 \neq \beta$ of site $s_\beta$,
	they form an $s_{\beta}$-cycle $C$ as shown in  Fig.~\ref{fig:unique_vld}.
	But $C$ is contained in $D_\pp$, deriving a contradiction 
	to Lemma~\ref{lem:nocycleindomain}.  
 \end{proof}

\begin{figure}
	\centering
	\includegraphics{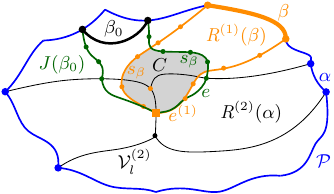}
	\caption{Illustrations for the proof of Theorem~\ref{thm:vldunique}.}
	\label{fig:unique_vld}
\end{figure}

\section{A randomized incremental algorithm}\label{sec:algorithm}

Consider a random permutation $o =(\alpha_1, \dots, \alpha_h )$ of the
set of core arcs  $\arcs$, where $|\arcs|=h$.
For $1 \leq i \leq h$, define set $\arcs_i = \{ \alpha_1, \dots, \alpha_i \} \subseteq \arcs$ 
to be the subset of the first $i$ arcs in~$o$, and permutation  $o_i=(\alpha_1,
\dots, \alpha_i )$.
Let  $\pp_i$ denote the boundary curve 
derived by the arc insertion operation $\oplus$ by considering arcs in  the order  $o_i$.
Let   $D_i$ denote the corresponding domain enclosed by $\pp_i$.

Our randomized algorithm 
is inspired by the randomized, two-phase, approach of Chew~\cite{Chew90} for the
Voronoi diagram  of
points in convex position;
however, the sites are core arcs in $\arcs$,
forming boundary curves, and the algorithm 
constructs Voronoi-like diagrams 
within a series of shrinking 
domains $D_{i}\supseteq D_{i+1}$.
The domain  $D_{1}$ is $D(s,s_{\alpha_1})\cap D_\Gamma$;
and $D_{h}$ coincides with the Voronoi region $\VR(s,S)\cap D_\Gamma$.
The boundary curves are obtained by the insertion operation $\oplus$,
one at each step, starting
with $\pp_1=J(s,s_{\alpha_1})\cap D_\Gamma$, and ending with $\pp_h=
\partial \VR(s,S) \cap D_\Gamma$. 
The algorithm works in two phases.  

In phase~1, the arcs in $\arcs$ get deleted
one by one, in  reverse order of $o$, while recording the neighbors  of
each arc at the time of its deletion.
Let  $\pp_1 = J(s,s_{\alpha_1})\cap D_\Gamma$,
$R(\alpha_1, \pp_1) = D(s,s_{\alpha_1})\cap D_\Gamma$, and $\vld(\pp_1)= \emptyset$.

In phase~2, we start with  $\vld(\pp_1)$ 
and incrementally compute $\vld(\pp_{i})$, $i =2, \dots, h$, by 
inserting arc $\alpha_{i}$, 
where  $\pp_{i} = \pp_{i-1} \oplus \alpha_{i}$,
and $\vld(\pp_{i})=\vld(\pp_{i-1})\oplus\alpha_{i}$.
When inserting an arc $\alpha_i$, we use the information of its recorded
neighbors from phase~1
to determine its insertion point. 
At the end we obtain $\vld(\pp_h)$, where $\pp_h$ is
a boundary curve of $\arcs$.
The set $\arcs$ has one unique boundary curve that coincides with its  $s$-envelope. Thus, $\pp_h$ can
contain no auxiliary arcs and $\pp_h=\env(\arcs)=\partial\VR(s,S) \cap D_\Gamma$.

We have already established  that the Voronoi-like diagram of an
$s$-envelope $\E$ is the real Voronoi diagram $\V(\E)$ (Corollary~\ref{lem:EVLD}).
We have also established the  correctness of the insertion operation
$\oplus$. 
Thus, the algorithm correctly computes  $\vld(\pp_h)$, where
$\vld(\pp_h)=\V(\arcs)=\V (S\setminus \{s\})\cap \VR(s,S)\cap
D_\Gamma$.

Next we analyze the time complexity of this algorithm and prove that
the time complexity of step-$i$ is expected $O(1)$; thus, the
overall time complexity is expected $O(h)$.

\begin{lemma}\label{lem:two-new-primes}
  $\pp_i$ contains  at most $i-1$ auxiliary arcs; thus,
  $|\vld(\pp_i)|=O(i)$.
\end{lemma}

\begin{proof}
  By definition, $|\pp_1| =1$.
	At each step of phase~2, exactly one original arc is inserted, and
        at most 
	one additional auxiliary arc is created by a split in
        case~(\ref{item:splitregion}) of
        Observation~\ref{obs:insert-beta}, except from $i=1$ and $i=h$. Thus, the 
        total number of auxiliary arcs is  at most $i-1$ and the
        number of original arcs is at most $i$.
        Since an original arc may be merged with its neighbor in case~(\ref{trivial}) of
        Observation~\ref{obs:insert-beta}, the number of original arcs
        in $\pp_i$ may indeed be less than $i$.
	Since the complexity of $\vld(\pp_i)$ is $O(|\pp_i|)$, the
        claim follows.
 \end{proof}

\subsection{Time analysis of the  randomized incremental algorithm}
\label{sec:time-analysis}

\begin{figure}
	\centering
	\includegraphics{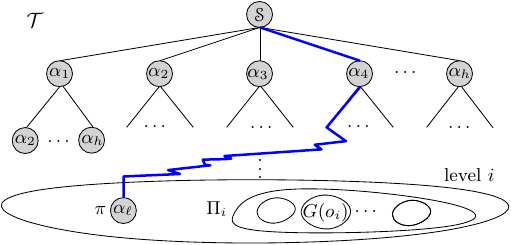}
	\caption{There are 
		$h!/(h-i)!$ nodes at level-$i$ of the decision tree~$\T$,
		each corresponding to a unique permutation of $i$ core
		arcs. Level $i$ 
is partitioned into groups of size $i$.}
	\label{fig:decision-tree}
\end{figure}

Consider the \emph{decision tree} $\T$ of all possible random choices
that can be made by our incremental algorithm
on the input set of core arcs   $\arcs$, $h=|\arcs|$,  see Fig.~\ref{fig:decision-tree}.
$\T$ has  $h!$ leaves each corresponding to one permutation of
the arcs in $\arcs$.
At level-$i$, there are $h!/(h-i)!$ nodes, and each node corresponds to a unique permutation of $i$ core
arcs.
A set of $i$ core arcs $\arcs_i$ is associated with $i!$ different nodes at
level-$i$, which are called the \emph{block of
  $\arcs_i$}. We have ${h \choose i}$ distinct such blocks at
level-$i$.
Although all nodes within one  block are associated with the same set of
core arcs, their corresponding boundary curves may vary considerably
depending on their permutation order. 
Because the boundary curves are order-dependent 
we cannot easily apply
\emph{backwards analysis} as in the original
randomized incremental construction of Chew~\cite{Chew90}.
Instead, we establish the expected complexity of step-$i$
by analyzing each block of  nodes at level-$i$
of~$\T$~\footnote{The analysis
  in our preliminary paper~\cite{JP18} follows the backwards analysis
  framework of Chew~\cite{Chew90},
  however,  the applicability is questionable because the boundary
  curves are order-dependent. We revisit and complete
  the analysis in this paper.}.

We use the following strategy.
We partition  each block at level-$i$ of  $\T$ into $(i-1)!$
disjoint  groups of $i$ nodes  each. For each group we show that
the step $i$ of the algorithm
requires total time $O(i)$, for all the $i$ permutations within the group.
Thus, on average, the algorithm spends $O(1)$ time on each node of $\T$.
Since all permutations are equally likely, we obtain
the expected linear ($O(h)$) time complexity of our algorithm.

Let $o_i= (\alpha_1, \alpha_2,  \dots, \alpha_i)$ be an 
arbitrary permutation of $\arcs_i$. 
From $o_i$ we define a group $G = G(o_i)$  of $i$ permutations: 
for each $1 \leq j < i$, 
remove ${\alpha_j}$ from its position in $o_i$ and append it to the end of $o_i$.
\begin{align}
o_i &= (\alpha_1, \alpha_2,  \dots, \alpha_{j-1}, \boxed{\alpha_j}, {\alpha_{j+1}}\dots, \alpha_{i-1}, \alpha_i)\\
o_j &= (\alpha_1, \alpha_2,  \dots, \alpha_{j-1}, \qquad \; {\alpha_{j+1}}, \dots,\alpha_{i-1},\alpha_i, \boxed{{\alpha_j}}), \label{eq:grouping}
\end{align}

\begin{figure}
	\centering
	\begin{minipage}{0.45\textwidth}
	\includegraphics{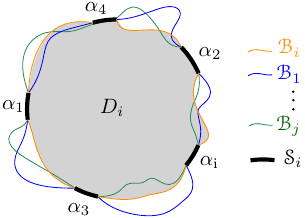}
\caption{Schematic differences between the boundary curves
	$\B_1,\dots,\B_i$.
	The domain $D_i$ is shown shaded.}
\label{fig:group}
	\end{minipage}
\hfill
	\begin{minipage}{0.45\textwidth}
		\centering
	\includegraphics{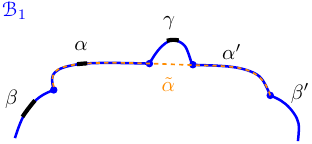}
	\caption{
	Illustration for Definition~\ref{def:source}, $o_1 = (\beta, \alpha,\gamma)$:
	The core arc $\alpha \in \arcs_i$ is the source of $\alpha' \in \In_1$.
	The expanded arc $\tilde \alpha \supseteq\alpha'$
	was created
	when inserting $\alpha$ during the  construction of $\B_1$. 
	$\B_i$ for $o_i = (\gamma, \beta, \alpha)$ is shown in
        Fig.~\ref{fig:primes}(a).
        }
\label{fig:source}
\end{minipage}
\end{figure}

Let $\B_j$, $1\leq j\leq i$, denote the boundary curves derived by
arc insertion 
following the order $o_j$, see
Fig.~\ref{fig:group}.
$\B_i$ is the base boundary curve derived from  $o_i$, and  its domain is denoted
$D_i$.
In the following we establish the relation between these boundary
curves so that we can prove our objective regarding the time
complexity of the
$i$th step of the algorithm on all of them (Lemma~\ref{lem:stepi}). 
We first introduce some terminology.

\begin{definition}
\label{def:source}
  Let $\alpha'$ be an auxiliary arc in $\B_j$ and let $\alpha\in
  \arcs_i$ be a core arc of the same site.
  We say that $\alpha'$ \emph{is an auxiliary arc of} the core arc 
    $\alpha$ if there had been an  expanded arc  
$\tilde \alpha \supseteq \alpha\cup\alpha'$, 
which was  created for the first time during the construction of $\B_j$
when inserting $\alpha$ (see
Fig.~\ref{fig:source}).
The core arc $\alpha\in \arcs_i$ is called 
the \emph{source of $\alpha'$} and is denoted as $\so_j(\alpha')$.

If $\alpha'$ appears  counterclockwise (resp. clockwise) from its source
$\alpha$ along their common $s$-bisector then $\alpha'$ is called 
a \emph{ccw} (resp. \emph{cw}) auxiliary arc.
\end{definition}

\begin{figure}
	\center
	\includegraphics{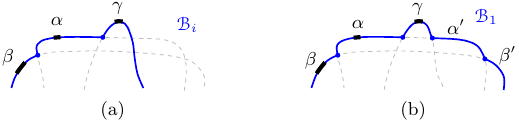}
	\caption{
		(a) Boundary curve $\B_i$, where 
		$o_i= (\gamma, \beta, \alpha)$. 
		(b) Boundary curve $\B_1$, where 
				$o_1= (\beta, \alpha, \gamma)$,
		containing	arcs $\alpha', \beta' \in \In_1$, 
		because ${\gamma}$ was inserted last.
	}
	\label{fig:primes}
\end{figure}

The boundary curves $\B_j$, $j<i$,  may get in and out of the domain $D_i$,
see Fig.~\ref{fig:group}.
To identify their differences from $\B_i$, let  $\In_j = \B_j \cap D_{i}$, and 
$\out_j = \B_j \setminus \overline{D_i}$,
denote the portion of $\B_j$  
inside, and outside of  $D_i$, respectively.
We partition the auxiliary arcs in $\In_j$ into $\In_j^+$  and
$\In_j^-$, where $\In_j^+$  (resp.  $\In_j^-$) includes  the  ccw  (resp. cw)
auxiliary arcs of $\In_j$, see Fig.~\ref{fig:primes}.
In the following we  only consider $\In_j^+$ as 
$\In_j^-$ is  symmetric.

\begin{observation}
\label{ob:out}
The boundary curve $\B_j$, $j\neq i$, contains no auxiliary arcs of the core arc
$\alpha_j$ (since
$\alpha_j$  appears last in  $o_j$), and 
these are the only auxiliary arcs of $\B_i$ that are  \emph{missing} from
$\B_j$ (since the insertion order of all other core arcs is identical).
Thus, any auxiliary arc $\alpha'\in \out_j$ must lie 
\emph{below} an
auxiliary arc 
of $\alpha_j$ in $\B_i$,
see Fig.~\ref{fig:primes-out}.
Further,
no region of an 
arc in $\out_j$ can be adjacent to $R(\alpha_j,\B_j)$.
\end{observation}

\begin{figure}[b]
	\center
	\includegraphics{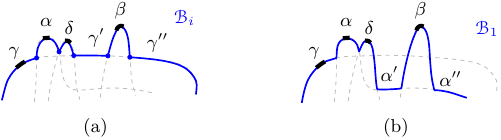}
	\caption{
		(a) Boundary curve  $\B_i$, where
				$o_i= (\gamma, \alpha, \beta, \delta)$. 
		(b) Boundary curve $\B_1$ containing 
		arcs $\alpha',\alpha''$ in $\out_1$, 
		where $o_1 = (\alpha, \beta, \delta, \gamma)$. 
	}
	\label{fig:primes-out}
\end{figure}

\begin{observation}\label{ob:+}
  Let $\alpha' \in \In_j$ and let $\alpha_k = \so_j(\alpha')$.
  Then $k>j$, i.e., $\alpha_k$ follows $\alpha_j$ in $o_i$.
  Further, if  $\alpha' \in \In_j^+$ then $(\alpha_k,\alpha_j,\alpha')$ appear ccw in
  $\B_j$. 
\deleted{
  Consider a ccw arc $\alpha' \in \In_j^+$ and let $\alpha_k$  be its source core arc 
	($\alpha_k = \so_j(\alpha')$). 
	Then $\alpha_k$ must follow $\alpha_j$ in $o_i$,
        i.e., $k>j$.
        Further, $(\alpha_k,\alpha_j,\alpha')$ must appear ccw in
        $\B_j$.
        Symmetrically for a cw arc in  $\In_j^-$.
        }
\end{observation}

\begin{observation}\label{obs:in}
  Fig.~\ref{fig:chainsmall} indicates the structure of $\In_j^+$.
  Let  $\alpha',\beta'\in \In_j^+$ such that  $\alpha_k =
  \so_j(\alpha')$, $\alpha_\ell = \so_j(\beta')$, and  $k<\ell$.
  Then  $j<k< \ell$ and $(\alpha_k, \alpha_\ell, \alpha_j,  \beta',
  \alpha')$ appear in ccw order 
  along $\B_j$, see Fig.~\ref{fig:chainsmall}.
 Further, all auxiliary arcs of $\alpha_\ell$ must appear before the
	auxiliary arcs of $\alpha_k$ as we move on  $\B_j$ counterclockwise from $\alpha_j$.
\end{observation}

\begin{figure}
	\centering
		\centering
		\includegraphics{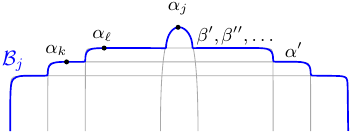} 
		\caption{	If $\alpha',\beta'  \in \In_j^+$, then
			$j<k< \ell$ and
			$(\alpha_k, \alpha_\ell, \alpha_j, \beta', \alpha')$ appear in ccw order on $\B_j$. }
		\label{fig:chainsmall}
\end{figure}

Since many auxiliary arcs of $\In_j^+$ can have the same source,
we define
\[\N_j=\{\so_j(\alpha') \in \arcs_i \;|\; \alpha' \in \In_j^+\}.\]

All arcs in $\N_j$ are of 
different sites.
Sets $\In^+_j$ and $\In^+_k $, $k \neq j$, may have many common arcs.
However, we have the following disjointness property.

\begin{lemma}\label{lem:Ndisjoint}
	$\N_j \cap \N_k = \emptyset$ for all $k \neq j$. 
	Thus, $\sum_{j =1}^i |N_j| = O(i)$.
\end{lemma}

\begin{proof}
	Suppose 
	$\alpha_\ell \in \N_j \cap \N_k$ and $j<k$,
	then $\alpha_\ell = \so_j(\alpha')$, where $\alpha' \in \In^+_j$ 
	and	
	$\alpha_\ell = \so_k(\alpha'')$, where $\alpha'' \in \In^+_k$. 
	(The arcs $\alpha'$ and $\alpha''$ may or may not
        overlap). 
        By Observation~\ref{ob:+},  $j<\ell$ (resp. $k<\ell$) and 
        $(\alpha_\ell,\alpha_j, \alpha')$ (resp.
        $(\alpha_\ell,\alpha_k, \alpha")$) must
        appear in  ccw  order on $\B_j$ (resp. $\B_k$).
        
        Suppose first that 
        $(\alpha_\ell,\alpha_k,\alpha_j)$
        appear in  ccw  order on $\B_i$.
        Then, since  $k<\ell$, the arc 
	$\alpha_k$ is inserted before $\alpha_\ell$ in $\B_j$,  and thus, 
	$\alpha'$ cannot  exist in $\B_j$, see
        Fig.~\ref{fig:Ndisjoint}. 
        Suppose now  that 
        $(\alpha_\ell,\alpha_j,\alpha_k)$
        appear in  ccw  order on $\B_i$.
        Then, since 
         $j<\ell$, the arc 
	$\alpha_j$ is inserted before $\alpha_\ell$ in $\B_k$,  thus, 
	$\alpha''$ cannot exist on $\B_k$, see
        Fig.~\ref{fig:Ndisjoint2}.
        In either case we derive a contradiction.
 \end{proof}

\begin{figure}[b]
	\centering
	\begin{minipage}{0.45\textwidth}
		\centering
		\includegraphics{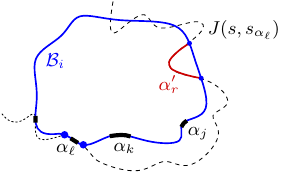}
		\caption{Illustration for Lemma~\ref{lem:Ndisjoint}.
			The case $(\alpha_\ell,\alpha_k,\alpha_j)$ appear ccw.}
		\label{fig:Ndisjoint}
	\end{minipage}
	\hfill
	\begin{minipage}{0.45\textwidth}
		\centering
		\includegraphics{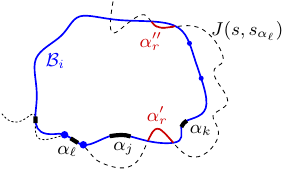}
		\caption{Illustration for Lemma~\ref{lem:Ndisjoint}.
			The case  $(\alpha_\ell,\alpha_j,\alpha_k)$ appear ccw.}
		\label{fig:Ndisjoint2}
	\end{minipage}
\end{figure}

We next establish that the parameters of the time complexity analysis
for step $i$, as given in Definition~\ref{def:parameters} and Lemma~\ref{lem:insertion-time},
sum up to $O(i)$ on all boundary
curves $\B_j, j\leq i$.

\begin{lemma}\label{lem:exp-const-primes}
 Considering all the boundary curves of group $G(o_i)$,
\[\sum_{j=1}^i 
(d_1(\alpha_j,\B_j) + 
d_2(\alpha_j,\B_j) +
\tilde d(\alpha_j,\B_j)) = O(i).\]

\end{lemma}

\begin{proof}
Let  $\aj$ and $\gj$ denote the original arcs preceding and
following $\alpha_j$ respectively in $\B_i$ (equiv. in $\B_j$).
Let 	
$d(\alpha_j,\B_k)$ 
denote
the auxiliary arcs on the boundary curve $\B_k$, $k=i,j$, from $\aj$
to $\gj$.

We first observe that $d(\alpha_j,\B_j)$ cannot contain any
portion of $\out_j$ because no auxiliary arc of $\alpha_j$ may appear in $\B_i$ from
$\aj$ to $\gj$, since $\alpha_j$ is the only core arc on  $\B_i$
between $\aj$ to $\gj$.
Thus, we only need to consider the auxiliary arcs of
$\In_j$. 
Next, we observe that no two auxiliary arcs in
$d(\alpha_j,\B_j)$ can have the same source in $ N_j$ for
the same reason, i.e., 
there is no core arc from $\aj$ to $\gj$ except $\alpha_j$.
Thus, we can bound  
$d(\alpha_j,\B_j) \leq    d(\alpha_j,\B_i)  + |N_j|$.
Then, by  Lemma~\ref{lem:Ndisjoint},
$\sum_{j=1}^id(\alpha_j,\B_j) \leq |\B_i|  + O(i) = O(i)$.
Since $d_1(\alpha_j,\B_j) + d_2(\alpha_j,\B_j) \leq d(\alpha_j,\B_j)$,
it follows $\sum_{j=1}^i(d_1(\alpha_j,\B_j) + d_2(\alpha_j,\B_j)) = O(i)$.

If $\tilde d(\alpha_j,\B_j)>0$, we have case~(\ref{item:splitgap}) 
of   Observation~\ref{obs:insert-beta}.
In this case, the endpoints of $\alpha_j$ are incident  to $\Gamma$, 
both in $\B_j$ and $\B_i$.
Then, by Observations~\ref{ob:out} and \ref{obs:in}, both $\In_j = \emptyset$
and $\out_j = \emptyset$, implying that $\B_j=\B_i$; thus, $\tilde d(\alpha_j,\B_j)= \tilde
d(\alpha_j,\B_i)$.
Then,
$\sum_{j=1}^i| \tilde d({\alpha_j},\B_j)| \leq |\tilde\B_i| = O(i)$.
 \end{proof}

\begin{figure}
	\centering
	\includegraphics{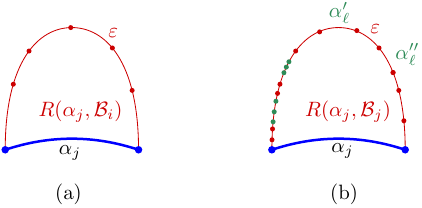}
	\caption{Illustration for Lemma~\ref{lem:adjacencies}.
	In between the two consecutively adjacent arcs
	$\alpha_\ell'$ and $\alpha_\ell''$ of  $\In_j$ of the same source 
	there must be an arc $\varepsilon \in \B_i$ that is 
	adjacent to $R(\alpha_j,\B_j)$.
	}
	\label{fig:comparison}
\end{figure}

\begin{lemma}\label{lem:adjacencies}
	$|R(\alpha_j,\B_j)| \leq 2 |R({\alpha_j},\B_i)| +  |N_j|$.
\end{lemma}

\begin{proof}
We compare $R(\alpha_j,\B_j)$ and $R({\alpha_j},\B_i)$ and bound 
differences in their adjacencies.
First, we observe that no arc in $\out_j$ 
can have a region adjacent 
to $R(\alpha_j,\B_j)$ (by Observation~\ref{ob:out}).
Next, we observe that any arcs common to both $\B_j$ and $\B_i$, whose
regions are
adjacent  to $R({\alpha_j},\B_i)$, they must also be adjacent to
$R({\alpha_j},\B_j)$.
In particular,  if  an arc $\varepsilon \in \B_j \cap \B_i$ has a
region $R(\varepsilon,\B_j)$ adjacent to
$R(\alpha_j,\B_j)$ then  $R(\varepsilon,\B_i)$
 must also be adjacent to $R({\alpha_j},\B_i)$.
This is clear, because otherwise, their common 
Voronoi edge $e$ in $\vld(\B_j)$ (or a portion of it) would be 
taken in $\vld(\B_i)$ by some arc in $\B_i$ that is \emph{missing} from $\B_j$,
by Lemma~\ref{lem:missingarc}). This must be an auxiliary arc
$\alpha_j'$ of  $\alpha_j$.
But if we insert  $\alpha_j'$ to  $\vld(B_j)$, 
the region $R(\alpha_j',B_j\oplus \alpha_j')$  will contain  a
portion of the  edge $e$, thus, it will be adjacent to 
$R(\alpha_j,B_j\oplus \alpha_j')$,
deriving a contradiction as arcs 
of the same site cannot be adjacent. 

Let $|R(\alpha_j,\B_j)|_x$ denote the number of additional adjacencies
that $R(\alpha_j,\B_j)$  may have over $R(\alpha_j,\B_i)$, i.e.,
$|R(\alpha_j,\B_j)| \leq |R(\alpha_j,\B_i)| + |R(\alpha_j,\B_j)|_x$.
We show that $|R(\alpha_j,\B_j)|_x\leq |R(\alpha_j,\B_i)| +|N_j|$.
Since auxiliary arcs of the same site can never have adjacent regions,
it follows that between
any two possible new adjacencies of $R(\alpha_j,\B_j)$ (counted in
$|R(\alpha_j,\B_j)|_x$) with auxiliary arcs of
$\In_j$ belonging to the same  source, there must be an 
adjacency with some arc not from $\In_j$, which 
by the first paragraph is  also contained in $\B_i$. 
Refer to Fig.~\ref{fig:comparison}(b), 
where in between the two consecutively adjacent arcs
$\alpha_\ell'$ and $\alpha_\ell''$ of  $\In_j$ 
the arc $\varepsilon \in \B_i$ is 
adjacent to $R(\alpha_j,\B_j)$.

Since by Observation~\ref{obs:in} auxiliary arcs of one source in
$N_j$ must appear in a certain order along $\B_j$ and they 
cannot alternate, the bound follows.
 \end{proof}

\begin{lemma}\label{lem:split}
  Consider case~(\ref{item:splitregion}) of
  Observation~\ref{obs:insert-beta} at the insertion time of $\alpha_j$ 
  in $\B_j$. Suppose that the insertion of 
  $\alpha_j$ splits an existing arc $\omega$ into two pieces
  $\omega_1$ and $\omega_2$.
 Then at least one of these two arcs 
  (say $\omega_1$) 
 must also exist in $\B_i$.
 	Further,  $|R({\omega_1},\B_j)| \leq 2 |R({\omega_1},\B_i)| +  |N_j|$.
\end{lemma}
\begin{proof}
        Suppose $\omega_1\alpha_j\omega_2$ appear in $\B_j$ in ccw
        order and 
	$\omega_2 \not\in \B_i$. Then $\omega_2 \in \In_j^+$, 
	see Fig.~\ref{fig:arcsplit}.
        Let $\alpha_\ell=\so_j(\omega_2)$, then  $\ell >j$ as
        $\omega_2 \in \In_j^+$.
        We claim that $\omega_1$ must
        belong in $\B_i$.

        Let  $\tilde\omega \supset \alpha_\ell$ denote the expanded arc
        created at the insertion time of $\alpha_\ell$ following the
        order  $o_j$.
        Clearly, $\tilde\omega \supset
        \omega$.
        Let  $\hat\omega\supset \alpha_\ell$ denote the expanded
        arc created  at the insertion time of $\alpha_\ell$,
        following $o_i$.
        Since $\ell >j$, it follows  that $\hat\omega$ can extend ccw at
        most until $\alpha_j$ and $\hat\omega\subset\tilde\omega$.
        Since  $\tilde \omega$ extends ccw past $\alpha_j$, it follows
        that no core arc $\alpha_\rho$, with $\rho<\ell$  can exist
        between $\alpha_l$ and $\alpha_j$. 
        Thus, $\hat\omega$ must extend ccw to $\alpha_j$ and 
        $\hat\omega \supset \omega_1$.
        In addition, no $\alpha_\rho$, with $\rho>\ell$,  can delete
        $\omega_1$ during its insertion, while following $o_i$, because the
        same would happen in $o_j$ and
        $\omega_1$ exists in $\B_j$.
        Thus, $\omega_1$ must exist in $\B_i$.

        We can now  bound 	$|R({\omega_1},\B_j)| \leq 2
        |R({\omega_1},\B_i)| +  |N_j|$ analogously to
        Lemma~\ref{lem:adjacencies}.
The only additional argument needed for the fact that 
no arc in $\out_j$ can have a region adjacent to $R(\omega_1,\B_j)$
is the observation that each arc in $\out_j$ 
lies below 
the $s_{\omega}$-bisector, 
because arc $\alpha_j$ splits arc $\omega$
(case~(\ref{item:splitregion}) of
Observation~\ref{obs:insert-beta}).
 \end{proof}
\begin{figure}
	\centering
	\includegraphics{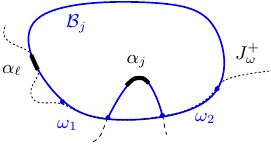}
	\caption{Illustration for the proof of Lemma~\ref{lem:split}. 
		If $\omega_2 \not\in \B_i$, then $\omega_1 \in \B_i$.}
	\label{fig:arcsplit}
\end{figure}

Let $T(i,o_j)$ denote the time that step-$i$ requires
following permutation $o_j$, i.e., the time required by the
last arc insertion of $o_j$.
\begin{lemma}\label{lem:stepi}
 The time for step-$i$  on the entire group
$G=G(o_i)$ is
 \[T(i,G)=  \sum_{o_j \in G} T(i,o_j) =O(i)\]
\end{lemma}

\begin{proof}
	Lemmas~\ref{lem:adjacencies} and \ref{lem:split} establish that
	$|R(\alpha_j,\B_j)| + |R(\omega_j,\B_j)|
	\leq 
	2(|R({\alpha_j},\B_i)| + |R(\omega_j,\B_i)| +  |N_j|)$, 
	where $\omega_j$ denotes one of the two arcs that is split and belongs to $\B_i$ 
	if case~(\ref{item:splitregion}) of Observation~\ref{obs:insert-beta} 
	is concerned.
Since $\omega_j$ is always an immediate neighbor of $\alpha_j$, 
we count it at most twice and thus, 
the total complexity  $\sum_{j=1}^i  |R(\omega_j,\B_i)|$ is $O(i)$.
	Together with Lemma~\ref{lem:Ndisjoint} this directly implies that 
	$\sum_{j=1}^i  |R(\alpha_j,\B_j)|  + r(\alpha_j,\B_j) = O(i)$.
	Lemma~\ref{lem:exp-const-primes} establishes that 
	$\sum_{j=1}^i 
	d_1(\alpha_j,\B_j) + 
	d_2(\alpha_j,\B_j) +
	\tilde d(\alpha_j,\B_j) = O(i)$.
        Then by Lemma~\ref{lem:insertion-time} the claim is derived.	
 \end{proof}

Before stating the final result, we show that the partitioning of each
block of $i!$ nodes (permutations)  at level-$i$ of $\T$
into $(i-1)!$ groups of $i$
permutations each,  is possible, if we follow the scheme we described 
in equation~\eqref{eq:grouping} for
$G(o_i)$.
Let  $\Pi_i$ denote such a block of all $i!$ permutations of the set
 $\arcs_i$.
The references and the proof of the following lemma were provided by Stefan
Felsner~\cite{stefan2019}.

\begin{lemma}\label{lem:grouping}
	The partitioning of $\Pi_i$ into groups by the scheme we
	defined 
	in equation~\eqref{eq:grouping} is possible, i.e.:
	For all $i \in \mathbb{N}$ and any block $\Pi_i$ of
        permutations on $\arcs_i$
	there exists a set 
	$F \subset \Pi_i$ of $(i-1)!$ permutations 
	such that 
	$\Pi_i = \dot\bigcup_{o\in F} G(o)$. 
\end{lemma}

\begin{proof}
	Following \cite{L92} denote by $\floor{\pi}$ the set of all permutations 
	that are obtained from a permutation $\pi$ by deleting one element.
	The following property is clearly an equivalent condition for a set $F$ 
	to satisfy $\Pi_i = \dot\bigcup_{o\in F} G(o)$.
	For each $\pi, \sigma \in F$ the sets  $\floor{\pi}$ and 
	$\floor{\sigma}$ are disjoint.  
	Levenshtein calls a family $F$ of $(i-1)!$ permutations with this 
	disjointness property a \emph{code capable of correcting single deletions} and proves 
	that these codes exist for all $i \in \mathbb{N}$ \cite[Theorem 3.1]{L92}.
 \end{proof}

All permutations at level-$i$ of the decision tree are
equally likely.
By Lemma~\ref{lem:grouping}, it is possible to partition them
into groups of $i$ nodes each, which  satisfy our scheme of
equation~\eqref{eq:grouping}.
By Lemma~\ref{lem:stepi}, each group requires total
$O(i)$ time to perform step $i$ on all its permutations.
We thus conclude: 

\begin{theorem}\label{thm:stepiconstant}
The expected time complexity of step $i$ of the randomized algorithm
is $O(1)$. 
\end{theorem}

We conclude with the following theorem.

\begin{theorem} \label{thm:deletion}
	Given an abstract Voronoi diagram $\V(S)$,  $\V (S\setminus
	\{s\})\cap \VR(s,S)$ can be computed in expected $O(h)$ time, 
	where $h$ is the complexity of $\partial \VR(s,S)$.
	Thus, $\V (S\setminus
	\{s\})$ can be updated from $\V(S)$ in expected time $O(h)$.
\end{theorem}

\section{Computing the order-$k$ Voronoi diagram iteratively}\label{sec:orderk}
Our algorithm to perform deletion in expected linear-time can be
adapted to iteratively compute
the order-${k}$ abstract Voronoi diagram, for increasing values of~$k$, 
in total time $O(k(n-k)n + n\log n)$ if $k\leq n/2$.
In particular, given a face~$\f$ of an order-$k$ Voronoi region, 
we can compute the order-$(k{+}1)$-subdivision
within~$\f$ in expected time $O(|\partial f|)$.
In this section we describe the required adaptation over site-deletion.

The \emph{order-$k$ abstract Voronoi region} of a  subset of sites $H \subset S$, $|H|=k$,
is defined~\cite{BCKLPZ15} as  
\[\kreg(H, S)~=~\bigcap_{q \in H,  p \in S \setminus H} D(q,p).\]

\noindent 
The  \emph{order-$k$ abstract Voronoi diagram} of $S$ is~\cite{BCKLPZ15}   
\[\Vk(S) = \mathbb{R}^2\setminus \bigcup_{H \subset S, |H| = k }\kreg(H, S).\] 

\noindent 
The combinatorial complexity of $\Vk(S)$ is $O(k(n-k))$.
For $k=1$,  it is the nearest-neighbor abstract Voronoi 
diagram $\V(S)$, and  for $k=n-1$, 
it is the farthest abstract Voronoi diagram $\fvd(S)$.
The vertices of the diagram  are  classified into \emph{new}
and \emph{old}, where a \emph{new} vertex in $\Vk(S)$ is an \emph{old}
vertex of $\V_{k+1}(S)$.

Consider a face $\f$  of an order-$k$ Voronoi region 
$\kreg(H)$, $H \subset S, |H|=k$.
Let $S_\f \subseteq S\setminus H$ denote the set of 
sites, which together with $H$, induce the Voronoi edges on the
boundary  $\partial \f$.
Our goal is to compute the  Voronoi diagram of $S\setminus H$ within $\f$, $\V(S_\f) \cap \f$,
in expected linear time, i.e., in time $O(|\partial\f|)$.
This diagram is a tree (or forest if $\f$ is unbounded) with properties analogous to
Lemma~\ref{lem:treelike} (see also \cite{BKL19}).
To extend Theorem~\ref{thm:deletion} from $k=1$ to an arbitrary $k$, 
there is a non-trivial challenge to overcome:
the complexity of $\partial f$ depends not only  on
$|S_\f|$ but also on $k$.
A direct application of our deletion algorithm would not result in a
linear-time scheme for non-constant $k$.

Consider a face $\f$ of $\kreg(H, S)$ and its boundary $\partial \f$.
We call any piece of $\partial f$  between two consecutive \emph{new} vertices,
an  \emph{order-$k$ arc}.
Such an arc does not have constant complexity but may contain
a sequence of old Voronoi vertices on  $\partial \f$.
In this section, let $\arcs$ denote the collection of the order-$k$ arcs
along the boundary of $\f$. 

An order-$k$ arc  $\alpha$ is a piece of a so-called \emph{Hausdorff bisector} between a site
$s_\alpha\in S_f$ and  $H$ (see, e.g.,~\cite{P04} for the definition of
concrete Hausdorff bisectors and the Hausdorff Voronoi diagram of point-clusters).
In abstract terms, the \emph{Hausdorff bisector} between a
site $s_\alpha\in S_f$ and  $H$ can be  defined as
\[J(s_\alpha, H) = \partial \freg(s_\alpha, H \cup \{s_\alpha\}),\]

\noindent
where $\freg(s,S')$ is 
the farthest Voronoi region of a site $s\in S'\subseteq S$,
$\freg(s,S') = \bigcap_{q \in S' \setminus \{s\}} D(q,s)$.

$J(s_\alpha,H)$ is an unbounded  Jordan curve dividing the plane in two parts;  
let $D(s_\alpha,H)=\freg(s_\alpha, H \cup \{s_\alpha\})$.
The complexity of $J(s_\alpha, H)$ is $\Theta(|H|)$, and this is  an
obstacle in directly applying our randomized linear time scheme.
It is possible to overcome this problem by considering relaxed Hausdorff
bisectors whose complexity depends solely on the order-$k$ arc,
and which define a series of
even larger shrinking  domains enclosing the face $\f$.

Let $H_\alpha\subseteq H$
be the subset of sites in $H$ that, together with $s_\alpha$, define
the edges and vertices along the arc $\alpha$.
Instead of  $J(s_\alpha,H)$, which is hard to compute,  we
consider the Hausdorff bisector  $J(s_\alpha,H_\alpha)$, 
where $\alpha\subseteq  J(s_\alpha,H_\alpha)$,  and  has
 complexity  $\Theta(|H_\alpha|)$. 
In fact, $\alpha\subseteq J(s_\alpha, \tilde
H_\alpha)$, for any $H_\alpha\subseteq \tilde H_\alpha\subseteq H$.
Let $|\alpha|$ denote the complexity of arc $\alpha$,  $|\alpha|= |H_\alpha|$.
We make use of the following property.

\begin{lemma}\label{lem:hausdorff}
	$J(s_\alpha,H) \subseteq
	\overline{D(s_\alpha,\tilde H_\alpha)}\subseteq
	\overline{D(s_\alpha,H_\alpha)}$, where $H_\alpha\subseteq
        \tilde H_\alpha\subseteq H$.
\end{lemma}

\begin{proof}
	Since $H_\alpha\subseteq H$, we have 
	\begin{align}
	D(s_\alpha,H) = \freg(s_\alpha, H \cup \{s_\alpha\}) 
	\subseteq \freg(s_\alpha, H_\alpha \cup \{s_\alpha\})
	= D(s_\alpha,H_\alpha). 
	\end{align}
	Thus, it holds $J(s_\alpha,H) = \partial {D(s_\alpha,H)} \subseteq \overline{D(s_\alpha,H_\alpha)}$.
	Analogously we can show the subset relation for $\tilde H_\alpha$.
 \end{proof}

It is now straightforward to adapt the algorithm of
Section~\ref{sec:algorithm}, using appropriate
Hausdorff bisectors that are derived by the order-$k$ arcs in $\arcs$,
in place of the $s$-related bisectors in the previous sections.
The complexity of each such Hausdorff bisector must be proportional to the
complexity of its underlying order-$k$ arc. 
Lemma~\ref{lem:hausdorff} implies the correctness of adopting this relaxation.

We start with domain $D_1$ defined by $J(s_{\alpha_1}, H_{\alpha_1})$,
i.e., $D_1=D(s_{\alpha_1}, H_{\alpha_1})\cap
D_\Gamma$, for the first
order-$k$ arc $\alpha_1$ of a random permutation of $\arcs$.
The boundary complexity of $D_1$ is $O(|\alpha_1|)$.
Note that $D_1$ is a superset of domain $D(s_{\alpha_1}, H)\cap
D_\Gamma$.
At step $i$, we insert arc $\alpha_i$ considering bisector
$J(s_{\alpha_i}, \tilde H_{\alpha_i})$, where  $H\supseteq\tilde
H_{\alpha_i}\supseteq H_{\alpha_i}$, and  $|\tilde
H_{\alpha_i}| \leq  |H_{\alpha_i}|+2$. 
We use $\tilde H_{\alpha_i}$, possibly a superset of $H_{\alpha_i}$,
in order to include at most one site in $H$ for each neighbor of $\alpha_i$
in $\pp_i$.
This is done to correctly link two neighboring order-$k$ arcs on
$\pp_i$ so that they are both incident to a common (new)
Voronoi vertex.
By Lemma~\ref{lem:hausdorff},  domain $D_i$ is a superset of the domain we would get if we
instead considered bisector $J(s_{\alpha_i}, H)\supset \alpha_i$.
Therefore,
the relaxed  construction works correctly.
At the end, $D_h=\f$.

We conclude that Theorem~\ref{thm:deletion}  applies, constructing $\V(\arcs)=\V(S_\f)\cap \f$ in expected
time $O(|\partial \f|)$.

Since the complexity of $\V_k(S)$ is $O(k(n-k))$,   the  $O(k^2(n-k) +
n\log n)$ bound for iteratively constructing the diagram, starting at
$\V(S)$,  easily follows  for $k\leq n/2$.
Although there are algorithms of better time-complexity  to
construct  $\V_k(S)$, such as the $O(k(n-k)\log^2n +
n\log^3 n)$ randomized incremental algorithm of Bohler et
al.~\cite{BKL19},
the iterative construction is nice and simple, 
therefore, it can be 
preferable for small values of~$k$.


\section{The farthest abstract Voronoi diagram}	 \label{sec:farthest}

In this section we show how to modify (in fact simplify) the algorithm for 
the deletion of one site 
to compute the \emph{farthest} 
abstract Voronoi diagram, after the sequence of its faces at infinity 
is known.

The \emph{farthest Voronoi region} of a site $p \in  S$ is 
$\freg (p,S) = \bigcap_{q \in S \setminus \{p\}} D(q,p)$
and the \emph{farthest abstract Voronoi diagram} of $S$ is
$\fvd(S) = \mathbb{R}^2\setminus \bigcup_{p \in S}\freg(p, S)$.
$\fvd(S)$ is a tree of complexity $O(n)$, however,
regions may be disconnected and a 
farthest Voronoi region may consist of $\Theta(n)$ disjoint
faces~\cite{AbstractFarthestVoronoi}.
Let $D^*(p,q)=D(q,p)$;
then $\freg (p,S) = \bigcap_{q \in S \setminus \{p\}}D^*(p,q) $.

Unless otherwise noted, we adopt the following
convention:
we reverse the labels of bisectors and use
$D^*(\cdot,\cdot)$, in the place of $D(\cdot,\cdot)$, in most definitions
and constructs of Sections~\ref{sec:problem}, \ref{sec:insert}.  
Under this convention the definition of e.g., a $p$-monotone path
remains the same but it uses $\partial \freg(p,\cdot)$ in the place of  
$\partial\VR(p,\cdot)$.
The corresponding arrangement of $p$-related bisectors $\jj{p}{S'}$,
$S'\subseteq S$, is considered with the labels of 
bisectors and their dominance regions
reversed from the original system $\J$.

Consider the enclosing curve $\Gamma$ as defined in
Section~\ref{sec:defs}, and let $\arcs$ 
be the sequence of arcs on $\Gamma$ derived by
$\Gamma\cap \fvd(S)$.
$\arcs$ represents the sequence of the farthest
Voronoi faces in $\fvd(S)$ at infinity.
The domain 
of computation is $D_\Gamma$.
For an arc $\alpha$ of $\arcs$, let $s_\alpha$ denote
the site in $S$ for which $\alpha \subset \freg(s_\alpha, S)$. 
With respect to site occurrences, $\arcs$ is a Davenport-Schinzel sequence
of order~2. 
$\arcs$ can be computed 
in time $O(n\log n)$ in a divide and conquer fashion,
similarly to computing the
\emph{hull} of a farthest segment Voronoi
diagram, see e.g.,~\cite{PD13}.

We treat the arcs in $\arcs$ as sites and compute
$\V(\arcs)=\fvd(S)\cap D_\Gamma$. Let
$\fs{\alpha}$ denote the face of $\fvd(S)\cap D_\Gamma$ incident to
$\alpha \in \arcs$, see Fig.~\ref{fig:fvd}. 
$\V(\arcs)$ is a tree 
 whose
leaves are the endpoints of the arcs in  $\arcs$.

\begin{figure}
	\centering
	\includegraphics{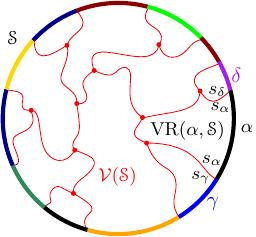}
	\caption{The farthest Voronoi diagram  $\V(\arcs) = \fvd(S)  \cap D_\Gamma$
		and the Voronoi region $\VR(\alpha, \arcs)$. 
		Bisector labels are shown in the farthest (reversed) sense. }
	\label{fig:fvd}
\end{figure}

Consider $\subs \subseteq \arcs$, and let $S'\subseteq S$ be the set of sites
that define the arcs in $\subs$.
Let $\J(S') = \{J(p,q) \in \J \,|\, p,q \in S', p \neq q\}$.

\begin{definition}
	A \emph{boundary curve $\pp$  for $\subs$}
	is a partitioning of $\Gamma$ into \emph{arcs}
	whose endpoints are in $\Gamma\cap \J(S')$ 
	such that any two consecutive arcs
	$\alpha,\beta \in \pp$  are incident to $J(s_\alpha,s_\beta)
	\in \J(S')$, having \emph{consistent labels}, and
	$\pp$ contains an arc $\alpha \supseteq \alpha^*$, 
	for every core arc $\alpha^* \in\subs$.
	We say that the labels of $\alpha$, $\beta$
	are consistent, if there is a neighborhood 
	$\tilde\alpha\subseteq\alpha$ and  
	$\tilde\beta\subseteq\beta$ incident to the common endpoint
	of $\alpha$ and $\beta$ such that $\tilde\alpha\in
	D^*(s_\alpha,s_\beta)$, and  $\tilde\beta\in
	D^*(s_\beta, s_\alpha)$. 
\end{definition}

There can be several different boundary curves for  $\subs$.
The arcs in $\pp$ that
contain a core arc in $\arcs'$ are called \emph{original} and any remaining arcs
are called \emph{auxiliary}. 
The arcs in $\pp$,  although they are arcs on
$\Gamma$, they  are all boundary arcs and none is considered a
$\Gamma$-arc in the sense of the previous sections.
The endpoint $J(s_\alpha,s_\beta)\cap\Gamma$ on $\pp$ separating two consecutive arcs 
$\alpha,\beta$  is denoted by 
$\nu(\alpha,\beta)$.

The Voronoi-like diagram of a boundary curve $\pp$ is defined 
analogously to Definition~\ref{def:vld}.
Since $\pp$ consists only of boundary arcs, $\vld(\pp)$ 
is a tree whose leaves are the
vertices of $\pp$.
The properties of a Voronoi-like diagram in Section~\ref{sec:problem} 
remain the same (under the conventions of  this section).

Given $\vld(\pp)$ for a boundary curve $\pp$ of $\subs  \subset
\arcs$, we can insert a  core arc  $\beta^*\in \arcs\setminus \subs$ and obtain 
$\vld(\pp\oplus \beta^*)$.
The insertion is performed analogously to Section~\ref{sec:insert}.
The original arc $\beta\supseteq\beta^*$, with endpoints $x,y$ is
defined as follows:
let $\delta$ be the first arc  on $\pp$ counterclockwise (resp. clockwise)
from $\beta^*$ such that $J(s_\beta, s_{\delta})\cap\delta\neq \emptyset$;
let  $x=\nu(\delta,\beta)$
(resp. $y=\nu(\beta,\delta)$).
Let $\pp_\beta=\pp\oplus \beta$ be the boundary curve obtained from
$\pp$ by substituting with $\beta$ its overlapping piece from $x$ to $y$.
No original arc of $\pp$ can be deleted by the insertion of $\beta$.
Observation~\ref{obs:insert-beta} remains the same, except from cases
(d),(e) that do not exist. 

The \emph{merge curve} $J(\beta)$, given $\vld(\pp)$, 
is defined analogously to Definition~\ref{def:mergecurve}; it  is
only simpler as it does not contain $\Gamma$-arcs.
Theorem~\ref{thm:mergecorrect} remains valid, i.e., 
$J(\beta)$ is an $s_\beta$-monotone path in $\jj{s_\beta}{S'}$ connecting the endpoints of $\beta$.
The proof structure is the same as for Theorem~\ref{thm:mergecorrect},
however, Lemma~\ref{lem:vertexonarc}  now requires 
a different proof, which we give in the
sequel (see Lemma~\ref{lem:fvd:gammabad}).
Lemma~\ref{lem:gammabad} is not
relevant; while 
Lemma~\ref{lem:finish} and Lemma~\ref{lem:distinctregions} 
are analogous.

In the following lemma we restore the
labeling of bisectors to the original. 

\begin{lemma}
	\label{lem:pcycles}
	In an admissible bisector system $\J$ (or  $\J\cup \Gamma$) there cannot be two
	$p$-cycles, $p\in S$,  with disjoint interior.
\end{lemma}

\begin{proof}
	By its definition, the nearest Voronoi region  $\VR(p,S)$
	 (resp. $\VR(p,S)\cap D_\Gamma$)
        must be enclosed 
	in the interior of any  $p$-cycle of the admissible bisector system $\J$
	(resp.  $\J\cup \Gamma$).
	But $\VR(p,S)$ (resp. $\VR(p,S)\cap D_\Gamma$)
	is connected (by axiom (A1)), thus, there cannot be two different $p$-cycles
	with  disjoint interior.
 \end{proof}

\begin{lemma}\label{lem:fvd:gammabad}
	Consider the merge curve $J(\beta)$. Suppose $v_{i+1}$ is not
	a valid vertex because $v_{i+1}\in \alpha_{i}$, i.e.,  $e_i$
	hits  arc $\alpha_i$.
	Then vertex $v_{m-j}$ can not be on $\pp$. 
\end{lemma}

\begin{proof}
	Suppose otherwise, i.e., vertex $v_{m-j}$ is on the boundary arc $\alpha_{m-j}$.
	Then $J_x^i $ and $J_y^j$ partition $D_\Gamma$
	in three parts: a middle part incident  to $\beta$, and two parts
	$C_1$ and $C_2$ at either side of  $J_x^i $ and $J_y^j$
	respectively, whose closures are disjoint,
	see Fig.~\ref{fig:fvd:hitboundary}.
	But the boundaries  of  $C_1$ and $C_2$ are 
	$s_\beta$-cycles in the admissible bisector system
	$\J\cup\Gamma$ contradicting Lemma~\ref{lem:pcycles}. Note
	that here we use the original labels of bisectors,
	including $\Gamma=J(s_\beta,s_\infty)$.
 \end{proof}

\begin{figure}
		\centering
		\includegraphics{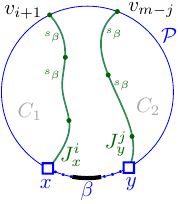}
		\caption{Illustration for Lemma \ref{lem:fvd:gammabad}.
			Nearest labels are shown.
		}
		\label{fig:fvd:hitboundary}
\end{figure}

The diagram $\ins$ is defined analogously
and the proof that $\ins$ is the Voronoi-like diagram  $\vld(\pp_\beta)$ for $\pp_\beta=\pp\oplus \beta$,
is analogous to the proof of Theorem~\ref{thm:insertion}.

The randomized algorithm for computing 
$\V(\arcs) = \fvd(S) \cap D_\Gamma$ 
is the same as in Section~\ref{sec:algorithm}.
The time analysis is also completely analogous.
For completeness we point out that, here, the set $\out_j$ consists of the auxiliary arcs in $\B_j$ that 
overlap with the auxiliary arcs of $\alpha_j$ in $\B_i$.
The set $\In_j$ are any remaining auxiliary arcs in $\B_j\setminus \out_j$ that 
differ from the corresponding auxiliary arcs in $\B_i$.
All observations of Section~\ref{sec:time-analysis} remain intact
under this updated notion of $\In_j$ and $\out_j$.
Thus, 
the (expected) linear time complexity can be analogously established.

\begin{theorem}\label{thm:FVDexp}
	Given the sequence of its faces at infinity, i.e., given
        the sequence of arcs $\arcs$ implied by  $ \fvd(S)\cap \Gamma$, 
      	the farthest abstract Voronoi diagram $\fvd(S)$ can be
        computed in expected linear time $O(|\arcs|)$. 
\end{theorem}

\section{Concluding remarks}

In this paper we formalized  the notion of an abstract \emph{Voronoi-like
  diagram}, which is defined as a graph (tree or forest) on the arrangement of the underlying 
bisector system whose vertices are legal Voronoi vertices
in Voronoi diagrams of  three sites.
We defined  the Voronoi-like diagram of a \emph{boundary curve}, which is  
implied by a subset  $\arcs'$ of Voronoi edges bounding 
a Voronoi region $\VR(s,S)$.
A boundary curve is defined as an $s$-monotone path in the arrangement of 
$s$-related bisectors that contains the arcs in $\arcs'$.
We showed that the Voronoi-like diagram of such a boundary curve is well-defined,
unique, and robust under
an arc-insertion operation, which enables its use in incremental
constructions.
Using Voronoi-like diagrams as intermediate structures, we derived a
very simple, therefore practical,
randomized incremental
algorithm to update an abstract Voronoi diagram  after
deletion of one site in expected  linear time.
The algorithm is applicable to any concrete diagram
under the umbrella of abstract Voronoi diagrams.

The technique can be adapted to compute the 
order-$(k{+}1)$ subdivision within an order-$k$
abstract Voronoi region, and the farthest abstract Voronoi diagram, after 
the order of its faces at infinity is known.
The Voronoi-like structure provides the means to efficiently deal with
the underlying disconnected Voronoi regions, which is
the common complication
characterizing these 
simple tree (or forest) Voronoi structures.

A deterministic linear-time construction  for these problems has remained
an open problem.
In future research, we would like to investigate the applicability of the
Voronoi-like structure within the 
 linear-time framework of Aggarwal\etal\cite{AGSS89} aiming to a
deterministic linear-time algorithm for the same problems.

\section*{Acknowledgements}
We sincerely thank Stefan Felsner for the proof of Lemma~\ref{lem:grouping} and
for making the connection to the seemingly unrelated result of Levenshtein
\cite{L92} on perfect codes, which established this claim for the time complexity analysis.

\bibliographystyle{plain}
\bibliography{voronoi}   

\end{document}